\documentclass[11pt]{article}

\usepackage{times} 

\makeatletter
\newcommand\subsubsubsection{\@startsection{paragraph}{4}{0ex}%
            {-3.25ex plus -1ex minus -0.2ex}%
            {1.5ex plus 0.2ex}%
            {\normalfont\normalsize\bfseries}}
\makeatother

\stepcounter{secnumdepth}
\stepcounter{tocdepth}

\usepackage{changepage}   
\usepackage{frcursive}
\usepackage{setspace}
\usepackage{sidecap}

\newcommand{\ETH}{{\sc Eth}\xspace}
\newcommand{\kSAT}{$k$-{\sc Sat}\xspace}

\usepackage{wrapfig}
\usepackage{colortbl}
\usepackage{etex}
\usepackage{bbding}
\usepackage{dingbat}
\usepackage[ruled]{algorithm2e}
\usepackage{amsmath}
\usepackage[newmultline]{empheq}
\usepackage{algorithmic}
\usepackage{amsfonts}
\usepackage{amssymb}
\usepackage{latexsym}
\usepackage{graphicx}
\usepackage{color}
\usepackage{balance}
\usepackage{textcomp}
\usepackage{nicefrac}
\usepackage{wasysym}
\usepackage[inline]{enumitem}
\usepackage{gensymb}
\usepackage{ctable}
\usepackage{fancybox}
\usepackage{longtable}
\usepackage[numbers,sort&compress]{natbib}

\usepackage{pifont,epsfig,rotating}
\usepackage{etoolbox}
\DeclareMathAlphabet{\mathpzc}{OT1}{pzc}{m}{it}
\usepackage{etaremune}
\usepackage{mathtools}

\definecolor{dgreyblue}{rgb}{0.26,0.3,0.46}             

\newcommand{\bee}{{\mathbf{e}}}
\newcommand{\cA}{\mathcal{A}}
\newcommand{\cB}{\mathcal{B}}
\newcommand{\cC}{\mathcal{C}}

\newcommand{\cE}{\mathcal{E}}
\newcommand{\cI}{\mathcal{I}}

\newcommand{\cO}{\mathcal{O}}

\newcommand{\cT}{\mathcal{T}}
\newcommand{\cU}{\mathcal{U}}

\newcommand{\cZ}{\mathcal{Z}}

\newcommand{\we}{{w_{\mathrm{edge}}}}
\newcommand{\wv}{{w_{\mathrm{node}}}}
\newcommand{\R}{{\mathbb R}}  
\newcommand{\PP}{{\mathbb P}}  

\renewcommand{\text}[1]{\hbox{\rm \ #1\ \/}}

\newcommand{\be}[1]{\begin{equation}\label{#1}}
\newcommand{\ee}{\end{equation}}
\newcommand{\beqn}{\begin{eqnarray*}}
\newcommand{\eeqn}{\end{eqnarray*}}
\newcommand{\beq}{\begin{eqnarray}}
\newcommand{\eeq}{\end{eqnarray}}
\newcommand{\ben}{\begin{enumerate}}
\newcommand{\een}{\end{enumerate}}
\newcommand{\bi}{\begin{itemize}}
\newcommand{\ei}{\end{itemize}}
\newcommand{\eps}{\varepsilon}

\newcommand{\IE}{{\em i.e.}\xspace}
\newcommand{\tx}{^{\rm th}}

\newtheorem{observation}{Observation}

\newtheorem{theorem}{Theorem}
\newtheorem{remark}{Remark}
\newtheorem{lemma}[theorem]{Lemma}
\newtheorem{corollary}[theorem]{Corollary}
\newtheorem{definition}[theorem]{Definition}

\newenvironment{proof}{{\noindent\bf Proof.\ }}{\hfill{\Pisymbol{pzd}{113}}\vspace{0.1in}}
\newenvironment{proof-sketch}{{\noindent\bf Sketch of Proof.\ }}{\hfill{\Pisymbol{pzd}{113}}\vspace{0.1in}}

\newcommand{\NP}{\mathsf{NP}}
\newcommand{\LP}{\mathsf{LP}}

\renewcommand{\deg}{\mathsf{deg}}

\newcommand{\cS}{\mathcal{S}}
\newcommand{\M}{\mathcal{M}}
\newcommand{\nbr}{\mathsf{Nbr}}
\newcommand{\cF}{\mathcal{F}}

\newcommand{\EA}{{\em et al.}\xspace}
\newcommand{\TB}{\vspace{-0.1ex}}\newcommand{\TiE}{\setlength{\itemsep}{-1ex}}

\newcommand{\junk}[1]{}
\newtheorem{proposition}{Proposition}

\newcommand{\EG}{{\it e.g.}\xspace}
\newcommand{\FI}[1]{Fig.~\ref{#1}\xspace}

\newcommand{\EC}{{\mathsf{E}_{\mathrm{critical}}}}
\newcommand{\ec}{{\mathsf{e}_{\mathrm{critical}}}}
\newcommand{\ENC}{{\mathsf{E}_{\mathrm{non-critical}}}}
\newcommand{\VC}{{\mathsf{V}_{\mathrm{critical}}}}
\newcommand{\VNC}{{\mathsf{V}_{\mathrm{non-critical}}}}

\newcommand{\OO}{\mathcal{O}}

\newcommand{\f}{{\mathfrak{f}}}

\newcommand{\msc}{{\sc Msc}\xspace}

\newcommand{\SCC}{{\sc Scc}\xspace}

\newcommand{\rrepeat}{{\bf{repeat}}}
\newcommand{\uuntil}{{\bf{until}}}

\newcommand{\dist}{\mathrm{dist}}
\newcommand{\vol}{\mathsf{vol}}

\newcommand{\eqdef}{\stackrel{\mathrm{def}}{=}}
\definecolor{columbiablue}{rgb}{0.61, 0.87, 1.0}

\newcommand{\ADHD}{{\sc Adhd}}

\newcommand{\adhd}{{\sc Adhd}}

\newcommand{\mcpm}[1]{{\mbox{\sc Mcpm}}\ensuremath{_#1}}
\newcommand{\tco}{{\sc Transport}-{\sc Cost}}
\newcommand{\eopt}{ {E}_{\mathrm{opt}} }
\newcommand{\opt}{\mathsf{opt}}
\usepackage[scr=esstix,cal=boondox]{mathalfa}

\newcommand{\mcov}{{\sc MaxCov}}
\newcommand{\abig}{\ensuremath{\mathfrak{a}}}
\newcommand{\bbig}{\ensuremath{\mathfrak{b}}}
\newcommand{\sgn}{\mathsf{sign}}
\newcommand{\wt}{{\fontfamily{qzc}\selectfont wt}}
\newcommand{\uw}{{\fontfamily{qzc}\selectfont uw}}
\newcommand{\rt}{{\fontfamily{qzc}\selectfont rt}}
\newcommand{\ut}{{\fontfamily{qzc}\selectfont ut}}
\newcommand{\del}{{\fontfamily{qzc}\selectfont del}}
\newcommand{\ins}{{\fontfamily{qzc}\selectfont ins}}
\newcommand{\ptn}{{\fontfamily{qzc}\selectfont ptn}}
\newcommand{\ntp}{{\fontfamily{qzc}\selectfont ntp}}

\allowdisplaybreaks

\usepackage[colorlinks=true, allcolors=blue]{hyperref}
\usepackage{emerald}
\usepackage{multirow}

\title{On Identifying Critical Network Edges via Analyzing Changes in Shapes (Curvatures)}

\author{
Bhaskar DasGupta \\
Department of Computer Science \\
University of Illinois Chicago \\
Chicago, IL 60607, USA \\
\emph{bdasgup@uic.edu}
\and 
Katie Kruzan \\
Department of Mathematics, Statistics \& Computer Science \\
University of Illinois Chicago \\
Chicago, IL 60607, USA\\
\emph{kkruza3@uic.edu}
}

\begin{document}

\maketitle

\begin{abstract}
In recent years extensions of manifold Ricci curvature to discrete combinatorial objects
such as graphs and hypergraphs
(popularly called as ``network shapes''), 
have found a plethora of applications in 
a wide spectrum of research areas
ranging over 
metabolic systems,
transcriptional regulatory networks,
protein-protein-interaction networks,
social networks
and 
brain networks
to 
deep learning models
but, in contrast, 
they have been looked at by relatively fewer
researchers in the algorithms and computational complexity community.
As an attempt to bring these network Ricci-curvature related problems under the lens
of computational complexity and foster further inter-disciplinary interactions,
we provide a formal framework for studying algorithmic and computational complexity issues 
for detecting critical edges in an undirected graph using Ollivier-Ricci curvatures 
and provide several algorithmic and inapproximability results for problems in this framework.
Our results show some interesting connections 
between 
our problems, 
the 
exact perfect matching and 
perfect matching blocker problems for bipartite graphs
and 
two well-known combinatorial packing/covering problems.
\end{abstract}

\tableofcontents

\section{Introduction}

Useful insights for many complex systems 
are often obtained by representing them as networks and 
analyzing them using graph-theoretic tools~\cite{DL16,Newman-book,Albert-Barabasi-2002}.
Typically, such networks may have so-called \emph{critical} 
(elementary) components 
which encode
some significant non-trivial properties of these networks. 
There is a rich history in finding critical components of networks 
dating back to quantifications of fault-tolerance or redundancy
in electronic circuits or routing networks, and more recently 
in the context of systems biology 
in quantifying redundancies in biological networks~\cite{KW96,TSE99,ADGGHPSS11}
and confirming the existence of central influential neighborhoods in biological networks~\cite{ADM14}.
Identifying critical components of networks are usually preceded by providing  details for 
the following items: 
\begin{description}[labelindent=0cm,leftmargin=0.3cm]
\item[Network model selection:] 
The network model used in this paper is the undirected (unweighted and edge-weighted) graph (network)
used in many real-world applications such as protein-protein-interaction networks, 
functional correlation networks of brain regions and social networks. 
\item[Definition of the atomic critical components:]
The atomic critical component used in this paper is an edge of the graph just like as in 
many recent applied research works such as~\cite{CATAD21,YEWJS23}.
\item[Network properties used to measure criticality:]
The specific network property used in this article is the ``(sign) change'' of the so-called
``local shape'' (discrete network Ricci curvature) at an edge of the network. 
\end{description}
Our own motivations for studying problems discussed in the article stem
from application to brain networks altered during to neuro-degenerative diseases 
or injuries, for example, for better understanding importances of specific connections involving 
the occipital region, the superior parietal lobule and the frontal operculum
of a functional connectivity network of brain in attention deficit hyperactivity disorder.
However the purposed of writing this article was not just to prove new computational
complexity results, but to bring network Ricci-curvature related problems under the \emph{lens
of computational complexity} and foster interaction between algorithmic community and 
researchers in other domains that use such curvatures in various applications.
Various problems involving Gromov hyperbolic network curvature, which was formulated much earlier, 
have indeed been investigated by the algorithms community extensively (see~\cite{CDA21}
for a brief survey on this).
In contrast, even though both Ollivier-Ricci and Forman-Ricci network curvatures 
are used extensively by practitioners in network science, biology and other communities 
and have even found their usage in \emph{deep learning} via 
curvature graph networks~\cite{Ye2020Curvature}, they have been looked at by relatively fewer
researchers in the algorithms community.

\subsection{Quantifying the Shape of a Network: Short Informal Overview}

Curvatures are very natural measures of anomaly of higher dimensional objects in 
mainstream physics and mathematics, \EG, in general relativity.
However, extending these measures to networks need to overcome many key
challenges since networks are discrete objects that do not necessarily have
an associated natural geometric embedding. There are several ways previous researchers
have attempted to formulate notions of curvatures of networks and other
combinatorial objects.

An earliest quantification of network shapes
is \emph{Gromov-hyperbolicity}
defined via properties of geodesic triangles or
$4$-node conditions or 
Gromov-product nodes~\cite{book99,G87}.
This measure assigns one scalar value to the entire network.
There is a large body of research works dealing with theoretical and empirical aspects of this measure
(\EG, see~\cite{ipl15,CCDDMV18,CDEHV08,ADM14}).

More recently, researchers have attempted to extend the concept of Ricci curvature for Riemannian 
manifolds to networks in two major ways. 
One of them is 
via extending Forman's discretization of Ricci curvature for 
(polyhedral or CW) complexes~\cite{F03}
to networks. 
Informally, the Forman-Ricci curvature is applied to networks by 
topologically associating components (sub-networks) of a given network with higher-dimensional 
objects. 
Although formulated relatively recently, there are already a number of 
papers regarding empirical application of 
this measure~\cite{Sree1,Weber17,DJY20,CATAD21}.
The other approach, which is used in this article, 
is by using Ollivier's discretization of Ricci curvature.
For an informal understanding, consider 
transporting a infinitesimally small ball centered at a point of a manifold 
along a specific direction 
to measure the ``distortion'' 
of that ball due to the shape of the surface.
Ollivier's discretization~\cite{Oll11,Oll09,Oll10,Oll07}
provides appropriate network-theoretic analogs of such concepts.
Intuitively, it tries to mimic the concept of the sectional curvature at a point 
in the manifold.
The 
Ollivier-Ricci curvature for networks 
has already found its way to a plethora of applications in 
biological and social networks (\EG, see~\cite{CATAD21,Ricci,EFLSJ20,LRSJ21}), and 
more recently both these curvatures have been extended to directed and undirected
hypergraphs~\cite{SAAB25,d5n8-y58m}. 
Both Forman-Ricci and Ollivier-Ricci curvatures yield a scalar (number) for each edge
and therefore may be considered as a measure of the ``local shape'' of the network at 
a given edge.

\subsection{Why Study Network Curvature Measures?}
\label{sec-why-shape}

Although studying network curvatures is mathematically intriguing, it is natural 
to ask if there are other valid reasons for such studies. 
We point out several reasons for this below.
\begin{enumerate}[label=$\triangleright$,leftmargin=*]
\item
Network curvature measures can encode 
\emph{non-trivial} topological properties
that are \emph{not} expressed by more 
established network-theoretic measures such as
degree distributions, clustering coefficients or 
betweenness centralities (\EG, see~\cite{ADM14,WJS16,WSJ16,Samal18,CATAD21}). 
\item
Most of the curvature measures (including the Ollivier-Ricci curvature used in this article) 
can mostly be computed in polynomial time,
as opposed to $\NP$-complete network measures such as 
\emph{cliques}~\cite{GJ79} or 
\emph{densest}-$k$-\emph{subgraphs}~\cite{GJ79}, or 
some types of community decompositions such as \emph{modularity maximization}~\cite{DD13}.
\item
Network curvature measures can explain many phenomena one frequently encounters in real network-theoretic applications,
such as 
\textbf{(\emph{i})}
paths mediating up- or down-regulation of a target node starting from the same regulator node in 
\emph{biological regulatory networks} often have many small crosstalk paths~\cite{ADM14} and
\textbf{(\emph{ii})}
existence of congestions in a node that is not a hub in \emph{traffic networks}~\cite{ADM14,JLBB11},
that are \emph{not} easily explained by other network measures.
\end{enumerate}
Discretizations of manifold Ricci curvature to networks and other discrete metric spaces 
have already found applications in a wide spectrum of research areas involving 
graphs and hypergraphs ranging over 
metabolic systems~\cite{SAAB25,LRSJ21},
transcriptional regulatory networks~\cite{EFLSJ20},
protein-protein-interaction networks~\cite{EFLSJ20},
social networks~\cite{SAAB25,CDR23,pmlr-v84-chien18a,d5n8-y58m,SJB19},
in neuroscience applications such as
comparing brain networks to study slowly progressing
brain diseases such as \emph{attention deficit hyperactivity disorder}~\cite{CATAD21} 
and \emph{autism spectrum disorder}~\cite{Ricci,Elumalai2021}, 
and in deep learning models~\cite{Ye2020Curvature}.

\subsection{Basic Notations and Terminologies}
\label{basic-notations}

Let $A=(V,E)$ be a given undirected (weighted or unweighted) graph.
$\overline{u,v}$ and $\dist_A(u,v)$ 
denote a \emph{shortest path} and 
the \emph{distance}
between the nodes $u$ and $v$ in $A$, respectively. 
For a node $u\in V$, 
$\nbr_A(u)=\{ v \,|\, \{u,v\}\in E\}$
and 
$\deg_A(u)= |\, \nbr_H(u) \,|$ 
denote 
the neighborhood
and the \emph{degree} of $u$,
respectively.
For 
a set of node pairs $E'$ satisfying $E'\cap E=\emptyset$,
$A^{+E'}$ denotes the graph $(V,E\cup E')$ obtained by adding to $A$ the edges corresponding to the 
node pairs in $E'$.
For 
a subset of edges $E'\subseteq E$,
$A^{-E'}$ denotes the graph $(V,E\setminus E')$ obtained by deleting from $A$ the edges in $E'$.
The cost $\zeta(M)$ of any matching $M$ in a edge-weighted graph is the sum of edge weights of the edges in $M$.
We refer to an edge of weight $i$ as an $i$-edge.
For two positive integers $x$ and $y$, LCM$(x,y)$ denotes the least common multiple of $x$ and $y$.
The notation $\mathrm{poly}(x)$ indicates a function that is $O(x^c)$ for some positive constant $c$.

\subsection*{Special notations used throughout the paper}

Throughout the rest of the paper, $G$ is the given input graph, $\ec=\{u,v\}$ is the edge of $G$
whose curvature is considered for change, and the quantities $\abig$ and $\bbig$ are defined as:
\begin{gather*}
\abig = \frac{ \mathrm{LCM} \left(1+\deg_G(u),1+\deg_G(v) \right) }{1+\deg_G(u) }, 
\hspace*{0.3in}
\bbig = \frac{ \mathrm{LCM} \left(1+\deg_G(u),1+\deg_G(v) \right) }{1+\deg_G(v) }
\end{gather*}
We note the following regarding the values of $\abig$ and $\bbig$\footnote{One may wonder why quantities 
involving 
$\deg_G(u)+1$ and $\deg_G(v)+1$
are of importance for algorithmic performance, as opposed to 
just $\deg_G(u)$ and $\deg_G(v)$. 
The reason for this is that probability distributions for computing Ricci curvatures are done 
are assigned on the ``closed'' neighborhoods of nodes (\IE, the node itself is included in the neighborhood).}:
\begin{enumerate}[label=$\triangleright$]
\item
\textbf{$\pmb{\bbig=1}$ does not necessarily imply $\pmb{\deg_G(u)=\deg_G(v)}$}, 
rather it implies that $\deg_G(u)+1=c\times (\deg_G(v)+1)$ for some positive integer $c\geq 1$, \EG, 
$\deg_G(u)+1=100 \times (\deg_G(v)+1)$.
\item
$\abig=\bbig=1$ is equivalent to 
$\deg_G(u)=\deg_G(v)$.
\end{enumerate}

\section{Quantifying Edge Criticality via Ollivier-Ricci Network Curvature}

In this section, we precisely define how we can quantify the criticality (importance) 
of an edge via 
Ollivier-Ricci curvature of graphs.
Although we use the quantification for 
Ollivier-Ricci curvatures, 
it can also be used in a straightforward manner for 
Forman-Ricci curvatures
or 
any other type of edge curvatures (\EG,
curvature using displacement entropy~\cite{Oll11,OllV12,BS10}).

\subsection{Ollivier-Ricci Curvature of Graphs}
\label{sec-emd}
\newcommand{\emd}{{\sc Emd}}
\newcommand{\ric}{{\sc Ric}}

First, suppose we are given an unweighted graph $G=(V,E)$.
Given 
an edge $\{u,v\}\in E$, 
\cite{Oll07,Oll09,Oll11}
uses the 
``earth mover's distance'' (\emd)
measure to define a ``course Ricci curvature'' 
\ric$(u,v)$ 
along the edge $e=\{u,v\}$
in the following manner.
Let $V_u$ and $V_v$ be the closed neighborhoods of $u$ and $v$, respectively, , \IE, $V_{u} \eqdef \{u\}\cup \nbr_G(u)$
and 
$V_{v} \eqdef \{v\}\cup \nbr_G(v)$.
Let the probability distributions $\PP_u$ and $\PP_v$ be uniform distributions
over the nodes in 
$V_u$ and 
$V_v$, 
respectively, \IE,
$
\forall \, v_1\in V_u : \,
\PP_u(v_1) 
= 
\frac{1}{1+\deg_G(u)}
$
and 
$
\forall \, v_2\in V_v: \,
\PP_v(v_2) 
= 
\frac{1}{1+\deg_G(v)}
$.
The 
Ollivier-Ricci curvature
is then defined as 
({\em cf}.\ \cite[Definition~3]{Oll09}):
\begin{gather*}
\text{\ric}_{\!\!G}(e) \eqdef
\text{\ric}_{\!\!G}(u,v) = 
1 - \text{\emd}_{\!\!G}(V_{u},V_v,\PP_u,\PP_v) 
\end{gather*}
where
the earth mover's distance 
$\text{\emd}_{\!\!G}(V_{u},V_v,\PP_u,\PP_v)$,
which can be computed in polynomial time,
is an optimal solution of the \emph{linear programming} ($\LP$) problem shown in \FI{f1}.
We can think of 
every number $\PP_u(v_1)$ as the maximum total amount of ``earth'' (dirt) at the ``source'' node $v_1$ that can be moved to other 
(destination) nodes, 
and every number $\PP_v(v_2)$ as the maximum total amount of earth the ``destination'' node $v_2$ can store in its storage.
The cost of transporting one unit of earth from source node $v_1$ to destination node $v_2$ is 
$\dist_G(v_1,v_2)$, and the goal is to satisfy the storage requirement of all nodes by moving earths as needed while
\emph{minimizing} the total transportation cost.
One can also think of the \emd\ solution as the distance between two probability distributions 
$\PP_u$ and $\PP_v$ on the set of nodes $V_u$ and $V_v$ based on the shortest-path metric on $G$.
The measure is extended for a graph $G$ with \emph{positive} edge weights by redefining $\dist_G(v_1,v_2)$ 
to be \emph{minimum total weight} of a path over all possible paths between $v_1$ and $v_2$ and using the modified equation:
\[
\textstyle
\text{\ric}_{\!\!G}(e) \eqdef
\text{\ric}_{\!\!G}(u,v) = 1 - \frac{ \text{\emd}_{\!\!G}(V_{u},V_v,\PP_u,\PP_v) } { \dist_G(u,v) }
\]
Since the value of 
$\text{\emd}_{\!\!G}(u,v)
\eqdef
\text{\emd}_{\!\!G}(V_{u},V_v,\PP_u,\PP_v)$
remains the same if we swap the two sets of source and destination nodes,
{\bf we assume throughout the paper that} 
$\pmb{\deg_G(u) \leq \deg_G(v)}$
in the calculation of 
$\text{\ric}_{\!\!G}(u,v)$.

\begin{figure}[htbp]
\centering
\begin{tabular}{ r l r}
\\
\toprule
\multicolumn{3}{c}{{\bf variables}: $z_{v_1,v_2}$ for every pair of nodes $v_1\in V_u,v_2\in V_v$} 
\\
[5pt]
  \emph{minimize} &  
	      $\sum_{v_1\in V_u}\sum_{v_2\in V_v} \dist_G(v_1,v_2)\, z_{v_1,v_2}$ 
	                &
        (* minimize total transportation cost *)
\\
[5pt]
  \emph{subject to} 
	 & 
   $\sum_{v_2\in V_v} z_{v_1,v_2} = \PP_u(v_1),\,\,$ for each $v_1\in V_u$
	 & 
  (* take from $v_1$ as much as it has *)
\\
[5pt]
	 & 
   $\sum_{v_1\in V_u} z_{v_1,v_2} = \PP_v(v_2),\,\,$ for each $v_2\in V_v$
	&
  (* ship to $v_2$ as much as it needs *)
\\
[5pt]
	& 
  $z_{v_1,v_2}\geq 0,\,\,$ for all $v_1\in V_u,v_2\in V_v$
	& 
\\
\bottomrule
\end{tabular}
\caption{\label{f1}$\LP$-formulation for \emd\ on the set of nodes $V_u$ and $V_v$ with $|V_1|\times |V_2|$ variables.
Comments are enclosed by (* and *). 
The value of the objective function 
corresponding to a \emph{valid} (not necessarily optimal) solution 
$\vec{z}=\{z_{v_1,v_2}\,|\, v_1\in V_u, v_2\in V_v\}$ 
that satisfy all the constraints is denoted by 
$\text{\tco}_{\!\!G}(\vec{z},V_u,V_v,\PP_u,\PP_v)$
(thus, 
$\text{\emd}_{\!\!G}(V_u,V_v,\PP_u,\PP_v)
= \min_{\vec{z}} \,\{ 
\text{\tco}_{\!\!G}(\vec{z},V_u,V_v,\PP_u,\PP_v)
\}$). 
}
\end{figure}   

\subsection{Quantifying Criticality of an Edge}

It is well-known that 
\textbf{positive curvatures and negative curvatures 
define very different topological worlds}~\cite{Oll11}; 
for example, 
positive 
Ollivier-Ricci curvature
allow rapid mixing of a natural Markov process based on the distributions 
$\PP_u$ and $\PP_v$~\cite{BD97}
and 
see~\cite{Gro91} 
for further extensive discussions on the sign and geometric meaning of curvature.
Based on this observation, we adopt the extremal anomaly detection framework used in the context of Gromov-hyperbolic 
curvature of undirected unweighted graphs in~\cite{DJY20} to measure the criticality of 
an edge in the following way.
For a given
edge $\ec=\{u,v\}\in E$,
we measure the 
criticality of the edge 
by the amount of topological alterations 
of $G$ that we need to perform to change the value of 
$\sgn( \! \text{\ric}_{\!\!G} (\ec) )$, 
where 
$ \sgn(x) = 1$ if $x>0$
and 
$ \sgn(x) = 0$ if $x\leq 0$\footnote{All the results will hold even if 
$\sgn(0)$ is taken to be equal to $1$ or if $\sgn(0)$ is declared undefined.}.
We have several variations of the problem depending on the following attributes:
\begin{description}[labelindent=0cm,leftmargin=0.3cm]
\item[Unweighted vs.\ positively edge-weighted graphs (\uw\ vs.\ \wt):]
For weighted graphs, we assume that the edge weights are integers selected from
$\{1,2,\dots,W\}$ where $W$ is polynomial in $\deg_G(v)$.
We do not allow an edge weight to be zero since 
$\text{\ric}_{\!\!G}(e)$ is undefined for this case and also to preserve 
the metric property needed for any curvature definition. 
Our justification of allowing only integers as opposed to (polynomial-size) rational numbers 
stems from the fact that 
$\text{\ric}_{\!\!G}(e)$ is
``scale-invariant'' (\IE, multiplying all edge weights of a graph by a fixed positive 
number does not change the value of 
$\text{\ric}_{\!\!G}(e)$), and thus if the edge weights are from the set 
$
\{\frac{a_1}{b_1},\dots,
\frac{a_k}{b_k}
\}
$
then as long as LCM$(b_1,\dots,b_k)$ is a polynomial-size number we can scale all the 
weights to be integers of polynomial size.
\item[Nature of topolical change: deletion vs.\ insertion of edges (\del\ vs.\ \ins):]
We allow only insertion or deletion of edges in the graph, but not both.
For weighted graphs, we assume that the weight of an inserted edge is 
of polynomial size. However, it is not difficult to see that we do not need edges 
arbitrarily large, namely if an inserted edge has a weight greater than the diameter 
of the given graph then the insertion of the edge has no effect on the shortest path
between any two nodes and thus there is no need to insert such an edge.
Assuming that the given graph has weights from the set 
$\{1,\dots,W\}$, 
this provides an upper bound of $(|V|-1)W$ on the weight of an inserted edge.
\item[Restricted vs.\ unrestricted candidate edges for insertion/deletion (\rt\ vs.\ \ut):]
We define two cases of permissible edges 
for insertion or deletion.
\begin{enumerate}[label=$\blacktriangleright$,leftmargin=*]
\item
When insertions of edges are allowed then 
for the restricted case only those new edges (not in $E$) 
can be added that whose 
one end-point is in $\nbr_G(u)\setminus\{v\}$
and 
the other end-point is in $\nbr_G(v)\setminus\{u\}$.
For the unrestricted case, any new edge can be added.
\item
If only deletion of edges are allowed then, 
for either case, the edge $\ec$ itself cannot be deleted. 
For the restricted case, only those edges $\{u',v'\}\in E$ 
for which $u',v'\notin\{u,v\}$
can be deleted.
For the unrestricted case, any edge in $E$, except $\ec$, can be deleted.
\end{enumerate}
Note that for restricted edge insertions/deletions the neighborhoods of the two nodes $u$ and $v$
remain the same, but that is not the case for unrestricted insertions/deletions.
\item[Nature of curvature change: positive-to-negative vs.\ negative-to-positive (\ptn\ vs.\ \ntp):]
The two cases correspond to whether the topological changes cause 
$\sgn( \! \text{\ric}_{\!\!} (\ec) )$ to change from $1$ to $0$ or from $0$ to $1$.
\end{description}
A version of the problem will be denoted by combining the various attribute values in 
the same order they were presented in the above discussion, \EG, 
Problem~\wt-\ut-\del-\ptn\
or 
Problem~\uw-\ut-\ins-\ntp.
It is easy to see that for unweighted graphs, insertions ({\em resp.}, deletions) 
of edges can only decrease ({\em resp.}, increase) the value of 
$\text{\ric}_{\!\!} (\ec)$.
Thus, the problems 
\uw-\ut-\del-\ptn,
\uw-\rt-\del-\ptn,
\uw-\ut-\ins-\ntp\
and 
\uw-\rt-\ins-\ntp\
have no feasible solutions and need not be studied. 
However, the situation is not so simple for weighted graphs even for restricted insertions/deletions. 
To see this, note that when we insert 
({\em resp.}, delete) 
an edge, it may decrease 
({\em resp.}, decrease) 
the distance values between many nodes but it may also change 
the value of $\dist_G(u,v)$ and therefore it is possible for the ratio
$\frac{ \text{\emd}_{\!\!G}(u,v) } { \dist_G(u,v) }$
to either increase or decrease.

\begin{observation}\label{obs-inswt}~\\
\textbf{\emph{(\emph{a})}}
Consider \emph{Problem}~\wt-\rt-\ins-\ntp\
and suppose that $w(u,v)\in\{1,2,3\}$.
Then the value of $\dist_G(u,v)$ remains the same no matter which edges are inserted
and therefore
insertion of any edge \emph{cannot} increase the value of 
$\frac{ \text{\emd}_{\!\!G}(u,v) } { \dist_G(u,v) }$.
Thus, in this case
we may assume that in any optimal solution 
if an edge between two nodes are inserted 
then the weight of the edge is $1$ since it provides the maximum 
decrease of the distance between nodes in the graph and consequently 
the maximum decrease in 
$\text{\emd}_{\!\!G}(u,v)$
\emph{(\emph{this is used in the proof of 
Theorem~{\rm\ref{thm-alg2-ins}}})}.

\smallskip
\noindent
\textbf{\emph{(\emph{b})}}
By the same reasoning as in 
\textbf{\emph{(\emph{a})}}, 
for \emph{Problem}~\wt-\ut-\ins-\ntp\
if $w(u,v)\in\{1,2\}$
then 
insertion of any edge cannot decrease the value of 
$\frac{ \text{\emd}_{\!\!G}(u,v) } { \dist_G(u,v) }$
since $w(u,v)$ does not change, and we may assume that in any optimal solution 
all inserted edges are of weight $1$
\emph{(\emph{this is used in the proof of 
Theorem~{\rm\ref{thm-hard-ins}}})}.
\end{observation}

\subsection{Justification of Our Quantification of Edge Criticality Measure}

In Section~\ref{sec-why-shape} we discussed network curvature measures are 
worth studying both theoretically and empirically 
since they
may encode 
non-trivial topological properties
that are not expressed by more 
established network-theoretic measures such as
degree distributions, 
clustering coefficients or 
betweenness centralities.
Network curvatures
have already found applications in a wide spectrum of research areas involving 
graphs and hypergraphs ranging over 
metabolic systems~\cite{SAAB25,LRSJ21},
transcriptional regulatory networks~\cite{EFLSJ20},
protein-protein-interaction networks~\cite{EFLSJ20},
social networks~\cite{SAAB25,CDR23,pmlr-v84-chien18a,d5n8-y58m,SJB19},
and 
brain networks~\cite{CATAD21,YEWJS23,Ricci}
to 
deep learning models for graph classification and other problems~\cite{Ye2020Curvature,grover2025curvgad,MILANO2026102879}. 
For this article, 
main motivations for studying edge criticality problems stem
from application to 
undirected \textbf{biological} (not social) networks 
such as functional correlation brain 
networks that gets altered during to neuro-degenerative diseases 
or injuries, \EG, for better understanding criticalities (importances) of specific connections involving 
the occipital region, the superior parietal lobule and the frontal operculum
of a functional connectivity network of brain in healthy humans \emph{vs}.\ 
patients suffering from attention deficit hyperactivity disorder~\cite{CATAD21}.
For these kinds of networks many established network measures, especially 
those motivated by applications to non-biological networks, do not work very well.
We also point out in passing that, as observed in past research works such as~\cite{CATAD21},
merely using the weight of the edge as a measure of its criticality does not perform well.

A general way for finding critical components of biological systems 
is via information-theoretic definitions of redundancy 
based on mutual-information contents (a component is considered 
critical if its redundancy is low)~\cite{TSE99,TSE94,TSE96}.
For example, the authors of~\cite{TSE99} consider system $X$ consisting of $n$ elements that
produces a set of outputs $\OO$ via a fixed connectivity matrix from a subset
of these elements, and 
defines the redundancy of $X$ 
as the difference between summed mutual information upon perturbation between 
perturbed bi-partitions of $X$ summed over all bipartition sizes
and 
the mutual information between the entire system and $\OO$ (\emph{cf}.\ \cite[Equation~{[3]}]{TSE99}).
However, a computational difficulty in applying such a
definition is that 
the number of possible bipartitions could be astronomically large even
for a modest size network.
Measures avoiding
inspecting all
bi-partitions were also proposed in~\cite{TSE99}, but the computational
complexities and accuracies 
of these measures on larger
networks are not well-explored yet.
More computationally feasible graph-theoretic approaches to circumvent this issue 
was proposed in~\cite{ADGGHPSS11}
but the methods there work only for unweighted directed networks.

An alternate approach, used in many research publications such 
as~\cite{CATAD21,Elumalai2021,YEWJS23},
is to treat the value of the curvature of the edge itself as a measure 
of its criticality.
However, as can be seen in past research works such as~\cite{CATAD21}
the curvature value still may not distinguish the criticality in all cases. 
As an example, the values of 
$\text{\ric}_{\!\!G}(e)$
are the same even for the two unweighted graphs shown in
\FI{fig2-cap}
even though the local topologies around those edges are very
different in the two graphs.

\sidecaptionvpos{figure}{c}
\begin{SCfigure}[50]
\scalebox{0.8}[0.8]{\includegraphics{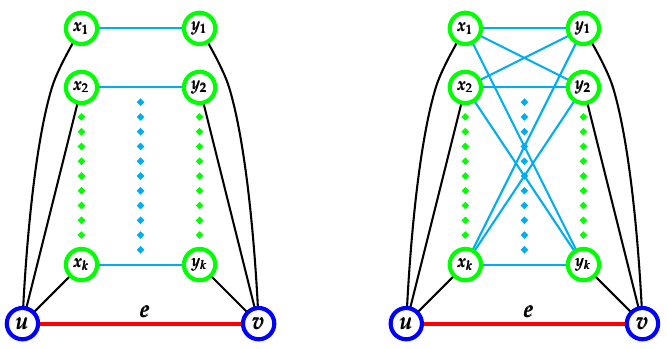}}
\caption{\label{fig2-cap}
Two graphs with the same value for 
$\text{\ric}_{\!\!G}(e)$
even though their local topologies are very different.
However, viewed in the framework of 
Problem~\uw-\rt-\del-\ptn\
the criticalities of the two edges are different since the graph in the right
needs deletion of at least $k-1$ more edges to change the value of 
$\sgn( \! \text{\ric}_{\!\!G} (e) )$. 
}
\end{SCfigure}

One way to overcome this limitation is via using Ricci flows and topological surgeries, as
was done in papers such as~\cite{SAAB25,d5n8-y58m,Weber17,NLLG19,LBL22,BMB23}. 
This approach some success in biological applications such as in 
finding cores in biochemical reaction systems\cite{SAAB25} and in mapping cellular 
trajectories~\cite{BMB23}. However, this may require a large number of iteration 
if the convergence rate of the flow is slow and since at each iteration one needs to
recompute the curvatures of all edges in the network the resulting process may take a
long time. 

An alternate approach, such as in~\cite{DJY20,EOCNDKT24,j.1467-8659.2009.01525.x}, 
would be to characterize the criticality of an edge by 
measuring the robustness of the curvature due to structural changes, like edge removal or noise.
\textbf{This is the approach followed in this article}.
Going back to the graphs in \FI{fig2-cap}, viewed in the framework of 
Problem~\uw-\rt-\del-\ptn\
the criticalities of the two edges are quite different, with the graph in the right
needing deletion of at least $k-1$ more edges to change the value of 
$\sgn( \! \text{\ric}_{\!\!G} (e) )$.

\subsection*{Justifications for the Edge Insertion/Deletion Model}

There are two ways one can change the topology of a network, by inserting/deleting one or more edges
or by inserting/deleting one or more nodes (for many network measures, including the curvature measure 
used in this article, deletion of a node can be equivalently seen as deletion of all edges incident 
on the node together as a group).
In this paper we add or delete edges from a network while keeping the node set the same.
This scenario captures a wide variety of applications 
such as inducing desired outcomes in disease-related biological networks via gene knockout~\cite{SWLXLAA11,ZA15},
inference of minimal biological networks from indirect experimental evidences or gene 
perturbation data~\cite{ADDKSZW07,ADDS,W02}, and 
functional correlation brain networks~\cite{CATAD21}, 
to name a few.
However, the node addition/deletion model or a mixture of node/edge addition/deletion model is also 
significant in many applications such as measuring fault-tolerance in computer networks; we leave investigations of these models 
as future research topics.

\subsection{Clarifying Remarks on the Time Complexity Issues}

In this article, we are interested in distinguishing 
between polynomial-time computations vs.\ computational 
intractability via $\NP$-completeness and inapproximability results,
so we do not concern ourselves with the precise nature of the polynomial
for polynomial-time algorithms. 
Thus, 
we utilize any (out of many possible) polynomial-time algorithm 
$\cO$ 
(\EG, Dijkstra's algorithm) 
that provides the distance between two specific nodes of a 
given (possibly edge-weighted) undirected graph on demand.
Our input size is the number of nodes $|V|$ of the given input graph
as is standard in published research works for time complexity calculations. 
Our results satisfy the following two criteria:  
\begin{enumerate}[label=$\triangleright$,leftmargin=*]
\item
Our polynomial-time algorithms will use 
$\mathrm{poly}(\deg_G(u)+\deg_G(v))$
time plus a 
$\mathrm{poly}(\deg_G(u)+\deg_G(v))$ 
calls to $\cO$.
Since 
$\deg_G(u)+\deg_G(v)<2|V|$
our algorithms will therefore run in 
$\mathrm{poly}(|V|)$ time, as desired. 
\item
For our computational hardness ($\NP$-completeness and inapproximability) 
results, all of our constructed hard instances will satisfy the condition that   
$\deg_G(u)+\deg_G(v)=\Theta(|V|)$, and therefore proving these hardness 
results assuming 
$\deg_G(u)+\deg_G(v)$ as input size also implies the same hardness result 
when $|V|$ is considered to be the input size.
\end{enumerate}

\subsection{Synopsis of Results and Proof Techniques}

\addtolength{\tabcolsep}{-3pt} 
\begin{table}[h]
{
\begin{tabular}{ c    c     c c c      c     c }
\toprule
							 & 
							 & 
            \multicolumn{3}{|c|}{Upper bound} & 
               \multicolumn{1}{c|}{Lower bound} & 
\\
\cline{3-4}
\cline{5-6}
								    \multicolumn{2}{c}{Problem} 
        &
				     \multicolumn{1}{|c}{
               \begin{tabular}{c}
						      \small Algorithm 
							 \\
										[-2pt]
								 \small Type$^{\!\!\text{\small\ding{172}}}$
               \end{tabular}
							 }
						      & 
               \begin{tabular}{c}
									\small Approximation 
									\\
										[-2pt]
									\small ratio
               \end{tabular}
						  &
				     \multicolumn{1}{c|}{
               \begin{tabular}{c}
									\small assumption 
									\\
										[-2pt]
									\small (if any)
               \end{tabular}
							 }
						  &
				     \multicolumn{1}{c|}{
               \begin{tabular}{c}
							  \small $\NP$-completeness 
								\\
										[-2pt]
								\small or
								\\
										[-2pt]
							  \small inapproximability 
								\\
										[-2pt]
							  \small ratio 
               \end{tabular}
							 }
              &
						      \begin{tabular}{c} \small Theorem, \\[-2pt] \small Section \end{tabular}
\\
\hline
\multirow{3}{*}{\small \wt-\rt-\ins-\ntp} & 
          \begin{tabular}{c}
					\small feasibility 
					\\
					[-2pt]
					\small checking
          \end{tabular}
					&
					\textsf{P}
					&
					---
					&
					\small
					$w(u,v)\in\{1,2,3\}$
					&
					---
					&
          \begin{tabular}{c}
            \small Theorem~\ref{thm-alg2-ins}\textbf{{(\emph{i})}}, 
						\\
						\small Section~\ref{sec-wt-rt-ins-ntp}
          \end{tabular}
\\
\cline{2-7}
 & 
          \begin{tabular}{c}
					\small finding 
					\\
					[-2pt]
					\small optimal
					\\
					[-2pt]
					solution
          \end{tabular}
					&
					---
					&
					---
					&
					---
					&
					 $\NP$-complete$^{\!\!\text{\small\ding{173}}}$
					&
          \begin{tabular}{c}
            \small Theorem~\ref{thm-alg2-ins}\textbf{{(\emph{ii})}}, 
						\\
						\small Section~\ref{sec-wt-rt-ins-ntp}
          \end{tabular}
\\
\midrule
\multirow{3}{*}{\small \uw-\rt-\ins-\ntp}  &
\begin{tabular}{c}
					\small feasibility$^{\dagger}$ 
					\\
					[-2pt]
					\small checking
          \end{tabular}
					&
					\textsf{P}
					&
					---
					&
					---
					&
					---
					&
          \begin{tabular}{c}
            \small Theorem~\ref{thm-alg-ins-feas}, 
						\\
						\small Section~\ref{sec-uw-rt-ins-ntp-feas}
          \end{tabular}
\\
\cline{2-7}
&
					\multirow{7}{*}{
          \begin{tabular}{c}
					\small finding 
					\\
					[-2pt]
					\small optimal
					\\
					[-2pt]
					solution
          \end{tabular}
					}
					&
				   \textsf{RP}	
					&
				   $\bbig$	
					&
					$\spadesuit$
					&
					---
					&
          \begin{tabular}{c}
            \small Theorem~\ref{thm-alg-ins-rest}\textbf{{(\emph{a})}}, 
						\\
						\small Section~\ref{sec-uw-rt-ins-ntp-rest}
          \end{tabular}
\\
[-2pt]
\cline{3-7}
&
					&
				   \textsf{P}	
					&
				   $2\bbig$	
					&
					$\spadesuit$
					&
					---
					&
          \begin{tabular}{c}
            \small Theorem~\ref{thm-alg-ins-rest}\textbf{{(\emph{b})}}, 
						\\
						\small Section~\ref{sec-uw-rt-ins-ntp-rest}
          \end{tabular}
\\
[-2pt]
\cline{3-7}
&
					&
				   \textsf{RP}	
					&
				   $\abig+\bbig$	
					&
					---
					&
					---
					&
          \begin{tabular}{c}
            \small Theorem~\ref{thm-alg-ins-unrest}\textbf{{(\emph{a})}}, 
						\\
						\small Section~\ref{sec-uw-rt-ins-ntp-unrest}
          \end{tabular}
\\
[-2pt]
\cline{3-7}
&
					&
				   \textsf{P}	
					&
				   $2 (\abig+\bbig )$	
					&
					---
					&
					---
					&
          \begin{tabular}{c}
            \small Theorem~\ref{thm-alg-ins-unrest}\textbf{{(\emph{b})}}, 
						\\
						\small Section~\ref{sec-uw-rt-ins-ntp-unrest}
          \end{tabular}
\\
\midrule
\multirow{3}{*}{\small \uw-\rt-\del-\ptn} & 
          \begin{tabular}{c}
					\small feasibility$^{\dagger}$ 
					\\
					[-2pt]
					\small checking
          \end{tabular}
					&
					\textsf{P}
					&
					---
					&
					\small
					---
					&
					---
					&
          \begin{tabular}{c}
            \small Theorem~\ref{thm-lower}\textbf{{(\emph{a})}}, 
						\\
						\small Section~\ref{sec-uw-rt-del-ptn}
          \end{tabular}
\\
\cline{2-7}
 & 
          \begin{tabular}{c}
					\small finding 
					\\
					[-2pt]
					\small optimal
					\\
					[-2pt]
					solution
          \end{tabular}
					&
					---
					&
					---
					&
					---
					&
					 $\NP$-complete$^{\!\!\text{\small\ding{174}}}$
					&
          \begin{tabular}{c}
            \small Theorem~\ref{thm-lower}\textbf{{(\emph{b})}}, 
						\\
						\small Section~\ref{sec-uw-rt-del-ptn}
          \end{tabular}
\\
\midrule
       {\small \wt-\ut-\ins-\ntp} & 
          \begin{tabular}{c}
					\small finding 
					\\
					[-2pt]
					\small optimal
					\\
					[-2pt]
					solution
          \end{tabular}
					&
					---
					&
					---
					&
					---
					&
					 $\eps\ln(\deg_G(v))^{\!\!\text{\small\ding{175}}}$
					&
          \begin{tabular}{c}
            \small Theorem~\ref{thm-hard-ins}, 
						\\
						\small Section~\ref{wt-ut-ins-ntp}
          \end{tabular}
\\
\bottomrule
\\
    \multicolumn{7}{l}{\hphantom{kkkkk}\small$\!\!^{\text{\small\ding{172}}}\!$\textsf{RP}: polynomial running time with success probability $1-o(1)$ }
\\
    \multicolumn{7}{l}{\hphantom{kkkkk}\small\hphantom{$\!\!^{\text{\small\ding{172}}}\!$}\textsf{P}: deterministic algorithm with polynomial running time}
\\
[2pt]
    \multicolumn{7}{l}{\hphantom{kkkkk}\small$\!\!^{\text{\small\ding{173}}}\!$result holds even if  
			all edge weights of $G$ are $1$, $2$ or $3$ and $\bbig=1$ }
\\
[2pt]
    \multicolumn{7}{l}{\hphantom{kkkkk}\small$\!\!^{\text{\small\ding{174}}}\!$result holds even if  
                       $\deg_G(u)=\deg_G(v)$, also implies $\NP$-completeness of \wt-\rt-\del-\ptn}
\\
[2pt]
    \multicolumn{7}{l}{\hphantom{kkkkk}\small$\!\!^{\text{\small\ding{175}}}\!$result holds even if  
                       $\bbig=1$, assuming \textsf{P}$\neq\!\NP$, $0<\eps<1$ is a fixed (specific) constant}
\\
[2pt]
     \multicolumn{7}{l}{\hphantom{kkkkk}\small$^{\dagger}$also provides sufficient conditions for feasiblity}
\\
[2pt]
     \multicolumn{7}{l}{\hphantom{kkkkk}\small$^{\spadesuit}$there are no edges connecting pairs of nodes in 
		  $\nbr_G(v)\setminus\{u\}$ or pairs of nodes in $\nbr_G(u)\setminus\{v\}$}
\\
\end{tabular}
}
\caption{\label{tab-sum}A summary of our algorithmic and inapproximability results,
\textbf{excluding the structural results in Section~\ref{sec-folklore}} 
(Theorem~\ref{thm-folklore}, Proposition~\ref{prop-rationalility}, Proposition~\ref{prop-struct-H}, and
Corollary~\ref{coro-struct-H}). 
Parameters $\abig$ and $\bbig$ are as defined in 
Section~{\rm\ref{basic-notations}}, 
\emph{\IE}, 
$\abig = \frac{ \mathrm{LCM} \left(1+\deg_G(u),1+\deg_G(v) \right) }{1+\deg_G(u) } $, 
$\bbig = \frac{ \mathrm{LCM} \left(1+\deg_G(u),1+\deg_G(v) \right) }{1+\deg_G(v) } $.
} 
\end{table}
\addtolength{\tabcolsep}{3pt} 

In this article, 
we provide a formal framework for studying algorithmic and computational complexity issues 
for detecting critical edges in an undirected graph using Ollivier-Ricci curvatures 
and
provide several (polynomial-time) algorithms and $\NP$-hardness/inapproximability
results 
for problems in this framework
in Sections~\ref{sec-folklore}--\ref{wt-ut-ins-ntp}
in the following manner:
\begin{enumerate}[label=$\triangleright$,leftmargin=*]
\item
In Section~\ref{sec-folklore} we state a folklore theorem (proof provided in 
Section~\ref{sec-app-folklore} of the appendix)
that relates the earth mover's distance calculation part in the Ricci curvature computation to 
finding a minimum-cost perfect matching, 
and subsequently prove useful results
(Theorem~\ref{thm-folklore}, Proposition~\ref{prop-rationalility}, Proposition~\ref{prop-struct-H}, and
Corollary~\ref{coro-struct-H})
relating edge insertions/deletions to a perfect matching 
in an associated bipartite graph.
\textbf{Our remaining results are succintly summarized in Table~\ref{tab-sum}}.
\item
Our algorithmic 
results in Theorem~\ref{thm-alg-ins-rest} and Theorem~\ref{thm-alg-ins-unrest}
bring out some interesting
connections between the edge insertion model
and versions of the so-called ``red-blue perfect matching problem''.
Randomized algorithms for 
the red-blue perfect matching problem
(as an extension of the perfect matching problem) 
was first studied via Tutte matrix 
by 
Mulmuley, Vazirani and Vazirani in $1987$~\cite{MVV87} 
and later extended to other similar problems such as exact perfect matching in 
publications such as~\cite{elma23,GKMT17}.
\item
Our $\NP$-completeness and approximation hardness results in 
Theorem~\ref{thm-alg2-ins} and 
Theorem~\ref{thm-hard-ins}
use reductions from well-known combinatorial covering and packing problems.
\item
Our computational hardness result for the restricted edge-deletion model 
in Theorem~\ref{thm-lower}
shows some interesting connection with the so-called 
``perfect matching blocker'' problem for bipartite graphs. 
The problem and its analogous version for
transversals
were first studied for maximum matching problem for general graphs 
in~\cite{ZENKLUSEN20094306,Kmk07}
and subsequently for perfect matchings in bipartite graphs~\cite{LACROIX201225}
with motivations from quite different application domains~\cite{Foltin02,10.1016/j.cie.2010.12.002}. 
\end{enumerate}

\section{A Folklore Structural Result for 
\emd}
\label{sec-folklore}

Let $H=(L,R,w)$ be an edge-weighted \emph{complete} bipartite graph where 
$L=\{u_1,\dots,u_r\}$ and 
$R=\{v_1,\dots,v_s\}$
are the left and right partition of nodes and 
$w(u_i,v_j)\geq 0$
is a \emph{non-negative} integer denoting the weight of the edge $\{u_i,v_j\}$. 
Consider the computation of 
$\text{\emd}_{\!\!H}(L,R,\PP_L,\PP_R)$
where $\PP_L$ and $\PP_R$ are uniform distributions on $L$ and $R$, respectively. 
Given such a $H$, we build a new complete bipartite graph 
$\widehat{H}=(\widehat{L},\widehat{R},\widehat{w})$
in the following manner. Let 
$q=\abig r=\bbig s$
where
$q$ is the \emph{least common multiple} of $r$ and $s$.
Then, 
for every node $u_i\in L$ we have $a$ nodes $u_{i,1},\dots,u_{i,\abig}$ in 
$\widehat{L}$, 
for every node $v_j\in R$ we have $\bbig$ nodes $v_{j,1},\dots,v_{j,\bbig}$ in 
$\widehat{R}$, 
and 
for every edge $\{u_i,v_j\}$ in $H$ 
we have $\abig\bbig$ edges $\big\{ \{ u_{i,k},v_{j,\ell}\} \,|\, k\in\{1,\dots,\abig\},\ell\in\{1,\dots,\bbig\} \big\}$ in 
with the weight of each of these $\abig\bbig$ edges being equal to 
$w(u_i,v_j)$.
Note that 
$\big|\widehat{L}\big|= \big|\widehat{R}\big|=q$.
Let 
$\widehat{\PP_L}$ and $\widehat{\PP_R}$ denote uniform distributions on $\widehat{L}$ and $\widehat{R}$, respectively,
and let 
\mcpm{{\widehat{H}}}
be the value of a minimum-cost perfect matching of $\widehat{H}$.
The following result relating 
$\text{\emd}_{\!\!H}(L,R,\PP_L,\PP_R)$
to 
$\text{\emd}_{\!\!\widehat{H}}(\widehat{L},\widehat{R},\widehat{\PP_L},\widehat{\PP_R})$
and 
a minimum-cost perfect matching of 
$\widehat{H}$
is implicit in many previous works in related combinatorial
optimization, but we could not find an explicit statement for it, so we 
state the result and its proof.

\begin{theorem}[Folklore]\label{thm-folklore}~\\[-20pt]
\begin{enumerate}[label={\textbf{({\alph*})}},leftmargin=*]
\item
From a valid solution 
$\vec{z}=\{z_{u_i,v_j}\,|\, u_i\in L, v_j\in R\}$ 
of value 
$\text{\tco}_{\!\!H}(\vec{z},L,R,\PP_L,\PP_R)$
on $H$ 
we can construct 
a valid solution 
$\vec{\widehat{z}}=\{\widehat{z}_{u_{i,k},v_{j,\ell}}\,|\, u_{i,k}\in \widehat{L}, v_{j,\ell}\in \widehat{R}\}$ 
of value 
\\
$\text{\tco}_{\!\!\widehat{H}}(\vec{\widehat{z}},\widehat{L},\widehat{R},\widehat{\PP_L},\widehat{\PP_R)}
=
\text{\tco}_{\!\!H}(\vec{z},L,R,\PP_L,\PP_R)$
on $\widehat{H}$, and vice versa.
\item
There exists
an optimal solution 
$\vec{\widehat{z}}=\{\widehat{z}_{u_{i,k},v_{j,\ell}}\,|\, u_{i,k}\in \widehat{L}, v_{j,\ell}\in \widehat{R}\}$ 
of value 
$\text{\emd}_{\!\!\widehat{H}}(\widehat{L},\widehat{R},\widehat{\PP_L},\widehat{\PP_R})$
on $\widehat{H}$
in which 
$\widehat{z}_{u_{i,k},v_{j,\ell}}\in\{0,\nicefrac{1}{q}\}$ for all $u_{i,k}\in \widehat{L}, v_{j,\ell}\in \widehat{R}$, 
and the set of edges $\big\{ \{u_{i,k},v_{j,\ell}\} \,|\, \widehat{z}_{u_{i,k},v_{j,\ell}}=\nicefrac{1}{q} \big\}$
of $\widehat{H}$ 
corresponds to 
a minimum-cost perfect matching 
of $\widehat{H}$ 
of cost 
$q \times \text{\emd}_{\!\!\widehat{H}}(\widehat{L},\widehat{R},\widehat{\PP_L},\widehat{\PP_R})$.
\item
A perfect matching $\M$
of $\widehat{H}$ 
of cost 
$\alpha$ correspond to a valid solution 
$\vec{\widehat{z}}=\{\widehat{z}_{u_{i,k},v_{j,\ell}}\in\{0,\nicefrac{1}{q}\}\,|\, u_{i,k}\in \widehat{L}, v_{j,\ell}\in \widehat{R}\}$ 
of value 
$\text{\tco}_{\!\!\widehat{H}}(\vec{\widehat{z}},\widehat{L},\widehat{R},\widehat{\PP_L},\widehat{\PP_R)}
=
\nicefrac{\alpha}{q}$
on $\widehat{H}$ 
in which 
$\widehat{z}_{u_{i,k},v_{j,\ell}}=\nicefrac{1}{q}$
if 
$\{u_{i,k},v_{j,\ell}\}\in\M$.
\end{enumerate}
The above three claims also imply
$\text{\emd}_{\!\!H}(L,R,\PP_L,\PP_R)=
\text{\emd}_{\!\!\widehat{H}}(\widehat{L},\widehat{R},\widehat{\PP_L},\widehat{\PP_R})=
\frac{\mcpm{{\widehat{H}}}}{q}$.
\end{theorem}

A proof of the above theorem can be found in Section~\ref{sec-app-folklore} of the appendix.

\subsection{Relating $\pmb{H}$ to \emd\ Calculation}
\label{sec-H-to-anomaly}

As before let 
$V_u=\{u_1,\dots,u_{r-2},u_{r-1}=u,u_r=v\}$
and 
$V_v=\{v_1,\dots,v_{s-2},v_{s-1}=u,v_s=v\}$
where 
$r=\deg_G(u)+1$
and 
$s=\deg_G(v)+1$,
Consider a complete edge-weighted bipartite graph 
$H=(L,R,w)$ 
where 
$L=V_u$,
$R=V_v$,
and 
$
w(u_i,v_j)=\dist_G(u_i,v_j)
$.
Note that 
$
w(u_i,v_j)=0 \equiv u_i=v_j
$,
$
w(u_i,v_j)=1 \equiv \{u_i,v_j\}\in E
$, and 
$
\text{\emd}_{\!\!G}(u,v)
=
\text{\emd}_{\!\!H}(L,R,\PP_L,\PP_R)
$.

\begin{proposition}\label{prop-rationalility}
The value of 
$\text{\emd}_{\!\!H}(L,R,\PP_L,\PP_R)$
is of the form
$\frac{\delta}{q}$ 
corresponding to a minimum-cost perfect matching $M$ of $\widehat{H}$ with $\zeta(M)=\delta$
for some non-negative integer $\delta$ 
satisfying $0\leq\delta<3q$ for unweighted graphs and
$0\leq \delta<3Wq$ for weighted graphs.
This implies that for unweighted graphs 
$
\text{\ric}_{\!\!G}(\ec)=
1-\text{\emd}_{\!\!H}(L,R,\PP_L,\PP_R)=
\frac{q-\delta}{q}
$
and for weighted graphs (since $\dist_G(u,v)\in\{1,2,\dots,W\}$) 
$
\text{\ric}_{\!\!G}(\ec)=
1-\frac{\text{\emd}_{\!\!H}(L,R,\PP_L,\PP_R)}{\dist_G(u,v)}=
1- \frac{\delta/\delta'}{q}
$
for some 
$\delta'\in\{1,2,\dots,W\}$.
\end{proposition}

\begin{proof}
Since
by Theorem~\ref{thm-folklore}(\emph{b}) 
there exists
an optimal solution 
$\vec{\widehat{z}}=\{\widehat{z}_{u_{i,k},v_{j,\ell}}\,|\, u_{i,k}\in \widehat{L}, v_{j,\ell}\in \widehat{R}\}$ 
of value 
$\text{\emd}_{\!\!\widehat{H}}(\widehat{L},\widehat{R},\widehat{\PP_L},\widehat{\PP_R})$
on $\widehat{H}$
in which 
$\widehat{z}_{u_{i,k},v_{j,\ell}}\in\{0,\nicefrac{1}{q}\}$ for all $u_{i,k}\in \widehat{L}, v_{j,\ell}\in \widehat{R}$
with 
$
z_{u_i,v_j} =
\sum_{k=1}^a \sum_{\ell=1}^b \widehat{z}_{u_{i,k},v_{j,\ell}}
$, 
we get 
$\text{\emd}_{\!\!H}(L,R,\PP_L,\PP_R)=
\delta\times \frac{1}{q}$ 
corresponding to a minimum-cost perfect matching $M$ of $\widehat{H}$ with $\zeta(M)=\delta$
for some integer $\delta>0$ and thus 
$
\text{\ric}_{\!\!G}(\ec)=
1- \text{\emd}_{\!\!H}(L,R,\PP_L,\PP_R)
=
\frac{q-\delta}{q}
$.
For weighted graphs (with edge-weights from $\{1,\dots,W\}$) 
the maximum weight of an edge between any node in $L$ and any node in $R$ 
is $3W$ and thus 
$0\leq \zeta(M)<3Wq$.
\end{proof}

\subsubsection{Relating Structural Properties of $H$ and $\widehat{H}$ to \emd\ Calculations for Unweighted Graphs}
\label{sec-common-nodes}

All graphs considered in this section are unweighted.
Let $\eta\geq 0$ be the number of common neighbors of $u$ and $v$,
and if $\eta>0$ let 
$
u_{r-2}=v_{s-2},
\dots,
u_{r-(\eta+1)}=v_{s-(\eta+1)}
\in \nbr_G(u)\cap \nbr_G(v)
$
be 
these common neighbors.

\begin{proposition}\label{prop-struct-H}
Given any minimum-cost perfect matching $\cT$ of $\widehat{H}$ we can 
modify it in polynomial time to another minimum-cost perfect matching 
$\cT'$ that satisfies the following property: 
%
\begin{quote}
for $j\in \{0,1,\dots,\eta+1 \}$,
$\cT'$ 
contains a
matching of some subset of 
$\bbig$ nodes from 
$\{u_{r-j,1},\dots,u_{r-j,\abig}\}$
to the $\bbig$ nodes
$v_{s-j,1},\dots,v_{s-j,\bbig}$ 
of total cost zero. 
\end{quote}
%
This implies that 
given any optimal solution of \emd\ on $H$ we can 
modify it in polynomial time to another optimal solution of \emd\ on $H$ that satisfies 
the following property: 
for every $j$
the modified optimal solution of the \emd\ on $H$ 
ships an amount 
of dirt
$
\frac{1}{1+\deg_G(v)}
$
from $u_{r-j}$ to $v_{s-j}$ with a total transportation cost of zero. 
\end{proposition}

\begin{proof}
Suppose for the sake of contradiction that 
$\kappa<\bbig$ nodes
from $\{v_{s-j,1},\dots,v_{s-j,\bbig}\}$
are matched to 
some subset of 
$\kappa$ nodes from 
$\{u_{r-j,1},\dots,u_{r-j,\abig}\}$.
This implies that there is a node in 
from $\{v_{s-j,1},\dots,v_{s-j,\bbig}\}$, say node 
$v_{s-j,1}$, 
that is not matched to any node in 
$\{u_{r-j,1},\dots,u_{r-j,\abig}\}$, 
and there is a node in 
$\{u_{r-j,1},\dots,u_{r-j,\abig}\}$, say the node 
$u_{r-j,1}$, that is not matched to any node in 
$\{v_{s-j,1},\dots,v_{s-j,\bbig}\}$.
Without loss of generality suppose that 
$u_{r-j,1}$ is matched to $v_{s-j',1}$ 
and 
$v_{s-j,1}$ is matched to $u_{r-j'',1}$
for some $j',j''\neq j$.
Note that 
$\widehat{w}(u_{r-j,1},v_{s-j,1})=0$
and
$\widehat{w}(u_{r-j,1},v_{s-j',1}),\widehat{w}(u_{r-j'',1},v_{s-j,1})\geq 1$. 
Since 
$\widehat{w}(u_{r-j'',1},v_{s-j',1})\leq 3$,
if 
$\max\{\widehat{w}(u_{r-j,1},v_{s-j',1}),\widehat{w}(u_{r-j'',1},v_{s-j,1})\}\geq 2$ 
then replacing the two edges 
$\{u_{r-j,1},v_{s-j',1})\}$ and $\{u_{r-j'',1},v_{s-j,1}\}$
by the two edges 
$\{u_{r-j,1},v_{s-j,1}\}$
and 
$\{u_{r-j'',1},v_{s-j',1}\}$
results in another perfect matching whose cost is no more than the given perfect matching.
So assume that 
$\widehat{w}(u_{r-j,1},v_{s-j',1})=\widehat{w}(u_{r-j'',1},v_{s-j,1})=1$. 
This implies the edges 
$\{u_{r-j},v_{s-j'}\}$ and $\{u_{r-j''},v_{s-j}\}$ 
exists in $G$; thus, since  
$u_{r-j}=v_{s-j}$ we have  
$\widehat{w}(u_{r-j'',1},v_{s-j',1})\leq 2$.
Then again 
replacing the two edges 
$\{u_{r-j,1},v_{s-j',1})\}$ and $\{u_{r-j'',1},v_{s-j,1}\}$
by the two edges 
$\{u_{r-j,1},v_{s-j,1}\}$
and 
$\{u_{r-j'',1},v_{s-j',1}\}$
results in another perfect matching whose cost is no more than the given perfect matching.
\end{proof}

\subsubsection{Edge Insertions/Deletions for Restricted Case 
(\uw/\wt-\rt-\ins-\ntp\ and \uw/\wt-\rt-\del-\ptn)}

For unweighted (respectively, weighted) graphs,
$
w(u_i,v_j)\in\{0,1,2,3\}
$
(respectively, 
$
w(u_i,v_j)\in\{0,1,\dots,3W\}
$)
and
adding a new edge $\{u_i,v_j\}$ with $u_i\in V_u\setminus\{u,v\},v_j\in V_v\setminus\{u,v\}$ in $G$ is equivalent to changing 
$w(u_i,v_j)\in\{2,3\}$ 
(respectively, $w(u_i,v_j)\in\{2,\dots,3W\}$) 
to $1$, 
and 
similarly deleting an existing edge $\{u_i,v_j\}$ of $G$ with $u_i\in V_u\setminus\{u,v\},v_j\in V_v\setminus\{u,v\}$ in $G$ 
is equivalent to changing 
$w(u_i,v_j)$ from $1$ to 
to a value greater than $1$ as given by shortest paths after edge deletion.
This allows us to equivalently work with the graphs $H$ and $\widehat{H}$ 
via using Theorem~\ref{thm-folklore} 
instead of $G$ when appropriate.
However, we need to be careful if we use the graphs 
$H$ or $\widehat{H}$ 
to depict these changes since the weight 
for an edge with an endpoint either $u$ or $v$
\emph{cannot} be decreased to $1$ in $H$ or $\widehat{H}$.
We can 
handle this situation
efficiently in the following way. 

\begin{definition}
The ``untouchable'' edges in $H$ correspond to those edges one of whose end-points is either
$u_{r-1}=u$ or $u_r=v$.
The ``untouchable'' edges in $\widehat{H}$ correspond to the edges 
one of whose endpoint is 
either
$u_{r-j,p}$
for $j\in\{0,1\}$ and $p\in\{1,\dots,a\}$  
or 
$v_{r-j,p}$
for $j\in\{0,1\}$ and $p\in\{1,\dots,\bbig\}$.   
The weights of these edges cannot be decreased
for the restricted case.
All other edges 
are ``touchable'' meaning that their weights may be decreased to $1$.
\end{definition}

For unweighted graphs, 
all untouchable edges have weights $2$; for weighted graphs 
all untouchable edges have weights from the set $\{2,3,\dots,2W\}$.

\begin{corollary}[Corollary of Proposition~\ref{prop-struct-H}]\label{coro-struct-H}
The matching $\cT'$ in 
Proposition~\ref{prop-struct-H}
contains exactly $2(\bbig-\abig)$ untouchable edges of weight $2$.
\end{corollary}

\section{Results for Restricted Edge Insertion/Deletion Model}
\label{sec-rt-ins-del}

\subsection{$\NP$-completeness of \wt-\rt-\ins-\ntp}
\label{sec-wt-rt-ins-ntp}

For reader's convenience, we write down 
Problem~\wt-\rt-\ins-\ntp\
precisely as follows. 
The input to the problem is an 
undirected weighted graph $G=(V,E,w)$ 
with each edge weight being an integer from the set $\{1,2,\dots,W\}$,
and 
an edge $\ec=\{u,v\}\in E$ such that $\text{\ric}_{\!\!G} (\ec)<0$.
A valid solution of the problem is a graph $G^{+E'}$ for a set of node pairs 
$E'=\{\{x,y\}\,|\, x\in \nbr_G(u)\setminus\{v\},y\in \nbr_G(v)\setminus\{u\}, x\neq y,\{x,y\}\notin E \}$
such that $\text{\ric}_{\!\!G^{+E'}} (\ec)>0$, and the 
objective is to \emph{minimize} $|E'|$.
We use the notation $\eopt(G,u,v)$ to denote an optimal solution of the problem 
of value $\opt(G,u,v)=|\eopt(G,u,v)|$.

\begin{theorem}\label{thm-alg2-ins}
Let $\bbig$ be as defined in 
Section~{\rm\ref{basic-notations}}, \emph{\IE},  
$\bbig = \frac{ \mathrm{LCM} \left(1+\deg_G(u),1+\deg_G(v) \right) }{1+\deg_G(v) }$.

\noindent
\textbf{\emph{(\emph{i})}}
If $w(u,v)\in\{1,2,3\}$ then 
there is a polynomial time algorithm to decide if 
{\em Problem}~\wt-\rt-\ins-\ntp\
has at least 
one valid solution.

\noindent
\textbf{\emph{(\emph{ii})}}
{\em Problem}~\wt-\rt-\ins-\ntp\
is $\NP$-complete even if all edge weights of $G$ are $1$, $2$ or $3$ and
$\bbig=1$.
\end{theorem}

\begin{remark}
Approximation algorithms \emph{(\emph{deterministic and randomized})} 
for 
Problem~\uw-\rt-\ins-\ntp\
for arbitrary $\bbig$ are described in 
{\rm Theorem~\ref{thm-alg-ins-rest}} and
{\rm Theorem~\ref{thm-alg-ins-unrest}}.
\end{remark}

\begin{remark}\label{rem-spades}
The graphs used in the $\NP$-completeness proof of 
Theorem~{\rm\ref{thm-alg2-ins}}
have the following property: ``there are no edges whose one end-point is in $\nbr_G(u)\setminus\{v\}$ and the other end-point 
is in $\nbr_G(v)\setminus\{u\}$, but there are many edges connecting pairs of nodes in $\nbr_G(v)\setminus\{u\}$''.
This leads to the natural question: is 
{\em Problem}~{\rm\wt-\rt-\ins-\ntp}
still computationally hard 
if the property was reversed, 
namely 
there are many edges whose one end-point is in $\nbr_G(u)\setminus\{u\}$ and the other end-point 
is in $\nbr_G(v)\setminus\{v\}$, but there are no edges connecting pairs of nodes in $\nbr_G(v)\setminus\{u\}$ or pairs of nodes in 
$\nbr_G(u)\setminus\{v\}$? 
We partially answer this question in the negative in 
Theorem~{\rm\ref{thm-alg-ins-rest}}
by showing that 
for unweighted graphs
there is a randomized polynomial-time $\bbig$-approximation and a deterministic polynomial-time 
$2\bbig$-approximation for this case, and thus when $\bbig=1$ 
\emph{(\emph{which is true for the $\NP$-completeness proof in 
Theorem~{\rm\ref{thm-alg2-ins}}})} 
we actually get a randomized algorithm solving the problem exactly.
\end{remark}

\noindent
\textbf{Proof of Theorem~\ref{thm-alg2-ins}}
{\em We will reuse the notations, observations and the proofs presented in Section}~\ref{sec-folklore} 
{\em as needed}.

\medskip
\noindent
\textbf{Proof of (\emph{i})}

\smallskip
If $w(u,v)\in\{1,2,3\}$, by  
Observation~\ref{obs-inswt}(\emph{a})
we may simply insert every possible edge (each with weight $1$) that is allowed
by the problem and check if this makes the curvature positive.

\medskip
\noindent
\textbf{Proof of (\emph{ii})}

\smallskip
Clearly 
{\em Problem}~\wt-\rt-\ins-\ntp\
is in $\NP$ since the curvature can be computed in polynomial time.
By Observation~\ref{obs-inswt}(\emph{a})
we may assume that every inserted edge has a weight of $1$.
We reduce the maximum coverage (\mcov) problem to our problem. 
The \mcov\ problem is as follows:
{\em given a universe $\cU=\{a_1,\dots,a_n\}$ consisting of $n$ elements and 
$m$ sets $\cS_1,\dots,\cS_m\subseteq\cU$ with $\cup_{j=1}^m \cS_j=\cU$
and an integer $\kappa>0$,
select $\kappa$ sets that maximizes the number of elements covered by them} (\IE, {\em maximizes number 
of elements in their union}). 
Let $\opt_{\kappa}$ denote the maximum number of elements covered by $\kappa$ sets. 
The following inapproximability result was proved by Feige~\cite{Feige98}
regarding \mcov: 
we can generate in polynomial time an instance 
$
\cI=
\langle
\cU,\cS_1,\dots,\cS_m,\kappa
\rangle
$
of \mcov\ such that for all constant $\eps>0$ 
the following claims hold: 
\textbf{(\emph{i})}
$|\cU|=n$ and $m\leq n^c$ for some integer $c\geq 1$, 
\textbf{(\emph{ii})}
either 
$\opt_{\kappa}=n$ 
or 
$\opt_{\kappa}\leq (1 - \nicefrac{1}{\bee} + \eps)n$, and  
\textbf{(\emph{iii})}
it is $\NP$-hard to decide if 
$\opt_{\kappa}=n$ 
or
$\opt_{\kappa}\leq (1 - \nicefrac{1}{\bee} + \eps)n$.
We modify this instance $\cI$ by replacing each element $a_i\in\cU$ 
by $n^{10c}$ elements $a_{i,1},\dots,a_{i,n^{10c}}$ in the universe 
and in each set that element appears. Note that this increases the 
universe size to $n^{11c}$ and if previously $\kappa$ sets 
covered $x$ elements they now cover $x\times n^{10c}$ elements. 
Also noting that $1-\frac{1}{\bee}<\frac{2}{3}-0.03$,
we rewrite the inapproximability result as follows: 
\begin{quote}
we can generate in polynomial time an instance 
$
\cI=
\langle
\cU,\cS_1,\dots,\cS_m,\kappa
\rangle
$
of \mcov\ 
with $\cU=\{a_{i,j}\,|\,i\in\{1,\dots,n\},j\in\{1,\dots,n^{10c}\} \}$
such that 
for $\eps=0.03$
the following claims hold: 
\textbf{(\emph{i})}
$|\cU|=n^{11c}$ and $m\leq n^c$ for some integer $c\geq 1$, 
\textbf{(\emph{ii})}
either 
$\opt_{\kappa}=n^{11c}$ 
or 
$\opt_{\kappa}\leq (\frac{2}{3}- \eps)n^{11c}$, and  
\textbf{(\emph{iii})}
it is $\NP$-hard to decide if 
$\opt_{\kappa}=n^{11c}$ 
or
$\opt_{\kappa}\leq (\frac{2}{3}- \eps)n^{11c}$.
\end{quote}
Given such an instance $\cI$ of \mcov\ as described above, we generate an instance 
$G=(V,E,w)$ 
of 
Problem \wt-\rt-\ins-\ntp\ in the following manner: 
\begin{enumerate}[label=$\triangleright$,leftmargin=*]
\item
In addition to $v$
node $u$ has a neighbor
$u_{\cT}$ with 
$w(u,v)=2$ and $w(u,u_{\cT})=3$.
We call $u_{\cT}$ as the ``solution neighbor (node)'' of $u$. 
\item
Other than $u$, node $v$ has the following $3n^{11c}+3m-2$ neighbors as follows.
\textbf{(\emph{a})}
There are $m$ neighbors 
$v_{\cS_1},\dots,v_{\cS_{m}}$ 
corresponding to the sets 
$\cS_1,\dots,\cS_{m}$
with 
$
w(v,v_{\cS_1}) = \dots =
w(v,v_{\cS_m}) = 
3
$;
we call these nodes as ``set neighbors (nodes)''.
\textbf{(\emph{b})}
Corresponding to each element $a_{i,j}\in\cU$ there is a
neighbor 
$v_{a_{i,j}}$
with 
$w(v,v_{a_{i,j}})=3$;
we call these nodes as ``element neighbors (nodes)''.
\textbf{(\emph{b})}
Node $v$ has 
additional 
$2n^{11c}+2m-2$ neighbors 
$v_1,\dots, v_{2n^{11c}+2m-2} $ with 
$ w(v,v_1)=\dots= w(v,v_{2n^{11c}+2m-2})=1$; 
we call these nodes as ``sink neighbors (nodes)''.
\item
For each element $a_{i,j}$ and set $\cS_\ell$ such that $a_{i,j}\in\cS_\ell$ we have a
``membership edge'' 
$\{v_{\cS_\ell}, v_{ a_{i,j} } \}$.
The weights of all membership edges are $1$.
\item
We add an edge of weight $1$ between every pair of set nodes. We refer to these as ``clique edges''. 
\end{enumerate}
It follows that 
$\deg_G(u)=2$,
$\deg_G(v)=3n^{11c}+3m-1$,
$q=
\mathrm{LCM}( 1+\deg_G(u),1+\deg_G(v) )
= 
3n^{11c}+3m
$, 
and 
$\bbig= \frac{ \mathrm{LCM} \left(1+\deg_G(u),1+\deg_G(v) \right) }{1+\deg_G(v) }=1$.
For notational convenience let $\delta=1+\deg_G(v)=
3n^{11c}+3m$;
thus the total number of set and element nodes is $\nicefrac{\delta}{3}$.
Note that no matter which new edges are inserted between the nodes in 
$\nbr_G(u)\setminus\{v\}$ and 
$\nbr_G(v)\setminus\{u\}$ under the restricted setting, the following are always true: 
\textbf{(\emph{a})}
$\dist_G(u,v)=2$, 
\textbf{(\emph{b})}
the distance between $u$ and any element or set node is at least $4$, 
\textbf{(\emph{c})}
the distance between $v$ and any element or set node is $3$,
\textbf{(\emph{d})}
the amount of dirt supplied by node $u_{\cT}$ is $\nicefrac{1}{3}$, 
and 
\textbf{(\emph{e})}
the sum of dirt demands of all set and element nodes is 
$\frac{\delta}{3}\times \frac{1}{\delta}=\frac{1}{3}$.
Also note that by 
Proposition~\ref{prop-struct-H}
we can assume for the \emd\ 
calculations that 
each of the nodes $u,v\in V_u$ has an amount 
$\frac{1}{3}-\frac{1}{\delta}$ 
of dirt for shipment 
and 
each of the nodes $u,v\in V_v$ has already received its demand of dirt completely.

We first show that 
$\text{\ric}_{\!\!G}(u,v)<0$.
The distance between any node in $\nbr_G(v)\setminus \{u\}$ from the nodes 
$u_{\cT}$, $u$ and $v$ are at least $6$, $3$ and $1$, respectively.
Thus,
$\text{\emd}_{\!\!G}(u,v)
>
6 \times \frac{1}{3} 
+ 3 \times (\frac{1}{3}-\frac{1}{\delta})
+ 1 \times (\frac{1}{3}-\frac{1}{\delta})
>2
$
which implies 
$\text{\ric}_{\!\!G}(u,v)=
1 - \frac{\text{\emd}_{\!\!G}(u,v)}{2}
<0
$.
Our proof will be completed by showing that 
\textbf{(\emph{a})}
if 
$\opt_{\kappa}=n^{11c}$ 
then 
$\opt(G,u,v) \leq \kappa$, and 
\textbf{(\emph{b})}
if 
$\opt_{\kappa}\leq (\frac{2}{3}- \eps)n^{11c}$
then 
$\opt(G,u,v) >\kappa$.

We first prove 
\textbf{(\emph{a})}.
Suppose that 
$\opt_{\kappa}=n^{11c}$ 
and the selected sets are $\cS_1,\dots,\cS_{\kappa}$.
We add the $\kappa$ edges 
$\{u_{\cT},v_{\cS_1}\},\dots,
\{u_{\cT},v_{\cS_{\kappa}}\}
$
resulting in a new graph $G'$.
Note that now the distance between the solution node and either a set node $v_{\cS_\ell}$ for $\ell>\kappa$ 
or an element node is $2$. 
Consider the following dirt shipment strategy. 
Node 
$u_{\cT}$
delivers the demands of all the set and element nodes fully,
and the nodes $u$ and $v$ delivers the demands of the sink nodes fully.
This results in a total shipment cost of 
$
1\times\frac{\kappa}{\delta} + 2 \times \left( \frac{1}{3}-\frac{\kappa}{\delta}\right)
+
3 \times \left( \frac{1}{3}-\frac{1}{\delta} \right) +
1 \times \left( \frac{1}{3}-\frac{1}{\delta} \right)
=
2-\frac{\kappa+4}{\delta}
<2
$
and thus 
$
\text{\ric}_{\!\!G'}(\ec)=
1- \frac{\text{\emd}_{\!\!G'}(u,v)}{2}
>0
$.

We now prove 
\textbf{(\emph{b})}.
Suppose that 
$\opt_{\kappa}\leq (\frac{2}{3}- \eps)n^{11c}$
but we still have been able to add $\kappa$ edges to make the curvature positive.
Let $G'$ be the new graph after edge insertions.
Suppose that we have inserted $\kappa_1$ edges between the solution node and the set nodes,
say the edges 
$\{u_{\cT},v_{\cS_{1}}\},\dots
\{u_{\cT},v_{\cS_{\kappa_1}}\}
$,
inserted $\kappa_2$ edges between the solution node and the element nodes,
say the edges 
$
\{ u_{\cT},v_{a_{j_1,\ell_1}} \},\dots
\{ u_{\cT},v_{a_{j_{\kappa_2},\ell_{\kappa_2}}} \}
$,
and 
inserted $\kappa_3$ edges between the solution node and the sink nodes,
say the edges 
$\{u_{\cT},v_{1}\},\dots
\{u_{\cT},v_{\kappa_3}\}
$,
with 
where $\kappa_1+\kappa_2+\kappa_3=\kappa\leq m\leq n^c$.
Noting that any set of $\kappa_1\leq\kappa$ sets can cover at most 
$(\frac{2}{3}- \eps)n^{11c}
=
(\frac{2}{3}- \eps)( \frac{\delta}{3}-m)
$
elements, the following statements hold:
\begin{description}
\item[(\emph{i})]
$\alpha_1\leq\kappa_1$ 
element nodes may receive their shipments of dirt by paths of length at most $1$, 
and 
$m-\alpha_1$
element nodes must receive their shipments of dirt 
via paths of length at least $2$.
\item[(\emph{ii})]
$\beta_1\leq \kappa_2$ 
element nodes may receive their shipments of dirt by paths of length at most $1$, 
$\beta_2\leq (\frac{2}{3}- \eps)( \frac{\delta}{3}-m)$
element nodes may receive their shipments of dirt by paths of length at most $2$, 
and 
$\frac{\delta}{3} - m - \beta_1 -\beta_2$
element nodes must receive their shipments of dirt 
via paths of length at least $3$.
\item[(\emph{iii})]
$\gamma_1\leq \kappa_3+  \delta\times ( \frac{1}{3}-\frac{1}{\delta} )$ 
sink nodes may receive their shipments of dirt by paths of length at most $1$, 
and 
$\frac{2}{3}\delta-\gamma_1$ 
sink nodes must receive their shipments of dirt by paths of length at least $3$. 
\end{description}
It thus follows that 
\begin{multline*}
\text{\emd}_{\!\!G'}(u,v)
\geq
\frac{
\alpha_1 + 2\times (m-\alpha_1) +
\beta_1 + 2\times \beta_2 + 
3\times ( \frac{\delta}{3} - m - \beta_1 -\beta_2 ) + 
\gamma_1 + 3 \times ( \frac{2}{3}\delta - \gamma_1 ) 
}
{\delta}
\\
=
\frac{
3 \delta  - \alpha_1 -m - 2\beta_1 - \beta_2 - 2\gamma_1
}
{\delta}
=
\frac{
  (\frac{19}{9}+\frac{\eps}{3} ) \delta - (\frac{1}{3}+\eps) m - (\kappa_1 + 2 \kappa_2 
	 + 
   2 \kappa_3 ) + 2 
}
{\delta}
>2
\end{multline*}
implying 
$
\text{\ric}_{\!\!G'}(\ec)=
1- \frac{\text{\emd}_{\!\!G'}(u,v)}{2}<0$, a contradiction.
\hfill{\Pisymbol{pzd}{113}}

\subsection{Polynomial Time Algorithms for Problem~\uw-\rt-\ins-\ntp}
\label{sec-uw-rt-ins-ntp}

Problem~\uw-\rt-\ins-\ntp\ is the same as 
Problem~\wt-\rt-\ins-\ntp\
as described in Section~\ref{sec-wt-rt-ins-ntp}
except that the graph $G$ is unweighted.
Let $\abig$ and $\bbig$ be as defined in 
Section~{\rm\ref{basic-notations}}, 
\emph{\IE}, 
\[
\abig = \frac{ \mathrm{LCM} \left(1+\deg_G(u),1+\deg_G(v) \right) }{1+\deg_G(u) }, 
\,\,\,
\bbig = \frac{ \mathrm{LCM} \left(1+\deg_G(u),1+\deg_G(v) \right) }{1+\deg_G(v) }
\]
Based on the discussion in Remark~\ref{rem-spades}, 
the following property for an instance of 
Problem~\uw-\rt-\ins-\ntp\
will be of interest for this section: 

\begin{adjustwidth}{0.6cm}{}
\textbf{$(\spadesuit)$} 
there are no edges connecting pairs of nodes in $\nbr_G(v)\setminus\{u\}$ or pairs of nodes in 
$\nbr_G(u)\setminus\{v\}$.
\end{adjustwidth}

\subsubsection{Checking Feasibility of Problem~\uw-\rt-\ins-\ntp}
\label{sec-uw-rt-ins-ntp-feas}

\begin{theorem}\label{thm-alg-ins-feas}
There is a polynomial time algorithm to decide if 
{\em Problem}~\uw-\rt-\ins-\ntp\
has at least 
one valid solution.
Moreover, 
if $\deg_G(v)<2 \deg_G(u)+1$
then there is always a valid solution.
\end{theorem}

\begin{proof}
We insert every possible edge, \IE,  
for every pair of nodes $x\in V_u\setminus\{u,v\},y\in V_v\setminus\{u,v\},x\neq y,\{x,y\}\notin E$ we add the 
edge $\{x,y\}$ to $G$, resulting in the graph $G'$.
We calculate 
$\text{\ric}_{\!\!G'} (\ec)=1-\text{\emd}_{\!\!G'}(u,v)$ 
and if it is positive then there is at least one valid solution.
We can calculate an upper bound for 
$\text{\emd}_{\!\!G'}(u,v)$ in the following manner.
In $G'$, 
we can ship 
an amount 
$
\frac{1}{1+\deg_G(v)}
$
of dirt
from $u\in V_u$ to $u\in V_v$ and
from $v\in V_u$ to $v\in V_v$
with a total transportation cost of $0$, 
an amount 
$
\frac{1}{1+\deg_G(u)}
-
\frac{1}{1+\deg_G(v)}
$
of dirt
from $u\in V_u$ to any subset of nodes in $V_v\setminus\{u,v\}$ 
with a total transportation cost of 
$
2 \times 
\left(
\frac{1}{1+\deg_G(u)}
-
\frac{1}{1+\deg_G(v)}
\right)
$, 
an amount 
$
\frac{1}{1+\deg_G(u)}
-
\frac{1}{1+\deg_G(v)}
$
of dirt
from $v\in V_u$ to any subset of nodes in $V_v\setminus\{u,v\}$ 
with a total transportation cost of 
$
2 \times 
\left(
\frac{1}{1+\deg_G(u)}
-
\frac{1}{1+\deg_G(v)}
\right)
$, 
and the rest of the transportations of dusts 
have an average cost of $1$ or less. 
This gives the following bound:
\begin{multline*}
\textstyle
\text{\emd}_{\!\!G'}(u,v) \leq 
4 \times 
\left(
\frac{1}{\deg_G(u)+1}
-
\frac{1}{\deg_G(v)+1}
\right)
+
\frac {\deg_G(u)-1} {\deg_G(u)+1}
=
1
+
\frac{2}{\deg_G(u)+1}
\,-\,
\frac{4}{\deg_G(v)+1}
\\
\textstyle
\Rightarrow 
\text{\ric}_{\!\!G'} (\ec)=1-\text{\emd}_{\!\!G'}(u,v)
\geq
\frac{4}{\deg_G(v)+1}
\,-\,
\frac{2}{\deg_G(u)+1}
\end{multline*}
Thus, 
$\text{\ric}_{\!\!G'} (\ec)>0$
provided $\deg_G(v)<2 \deg_G(u)+1$.
\end{proof}

\subsubsection{Polynomial Time Algorithms for Restricted Instances of Problem~\uw-\rt-\ins-\ntp}
\label{sec-uw-rt-ins-ntp-rest}

\begin{theorem}\label{thm-alg-ins-rest}
Assume that the given instance of 
{\em Problem}~{\rm\uw-\rt-\ins-\ntp}
has at least one valid solution and satisfies the property 
\textbf{$(\spadesuit)$}. 
Then the following claims are true.
\begin{description}
\item[\textbf{\emph{(a)}}]
There is a randomized polynomial time algorithm that returns a valid solution $E'$ of 
{\em Problem}~\uw-\rt-\ins-\ntp\
with a success probability of $1-o(1)$
satisfying 
$|E'|\leq \bbig\times \opt(G,u,v)$. 
\item[\textbf{\emph{(b)}}]
There is a deterministic polynomial time algorithm that returns a valid solution $E'$ of 
{\em Problem}~\uw-\rt-\ins-\ntp\
satisfying 
$|E'|\leq 2\bbig\times \opt(G,u,v)$.
For $\bbig=1$ we also exhibit an example for which 
our algorithm selects 
$2\times \opt(G,u,v)$ edges.
\end{description}
\end{theorem}

\begin{corollary}
If $1+\deg_G(v)=c\times (1+\deg_G(u))$ for some integer $c\geq 1$ then $\bbig=1$, and thus for this case 
{\em Problem}~\uw-\rt-\ins-\ntp\
under property 
\textbf{$(\spadesuit)$} 
admits a randomized polynomial time exact solution 
$(${\em \IE}, {\em Problem}~\uw-\rt-\ins-\ntp\ belongs to the complexity class \emph{RP}$)$
and 
admits a deterministic $2$-approximation.
\end{corollary}

\noindent
\textbf{Proof of Theorem~\ref{thm-alg-ins-rest}}
{\em We reuse the notations, observations and the proofs presented in Section}~\ref{sec-folklore}.
Since there are no edges connecting pairs of nodes in $\nbr_G(v)\setminus\{u\}$ or pairs of nodes in 
$\nbr_G(u)\setminus\{v\}$, insertion of an $\{u_i,v_j\}$ edge in $G$ correspond to change of 
$w(v_i,v_j)$ in $H$ from either $2$ or $3$ to $1$ and changes the weight of no other edge in $H$.
This change corresponds to 
changing 
the values of the edge weights $\widehat{w}(u_{i,k},v_{j,\ell})$
for $k\in\{1,\dots,\abig\},\ell\in\{1,\dots,\bbig\}$ 
in $\widehat{H}$ from $2$ or $3$ to $1$ and does not affect the weight of any other edge in 
$\widehat{H}$.

\smallskip
\noindent
\textbf{(\emph{a})}
Consider the graphs $H$ and $\widehat{H}$ via associations as described in 
Section~\ref{sec-H-to-anomaly}.
Our problem then equivalently can be expressed as:
\begin{quote}
find a minimum number of \emph{touchable} $2$-edges or $3$-edges 
such that the value of \emd\ on the 
graph obtained from $H$ by decreasing the weights of these edges to $1$ is strictly less than $1$.
\end{quote}
From $H$ we build the complete bipartite edge-weighted graph 
$\widehat{H}=(\widehat{L},\widehat{R},\widehat{w})$ 
as described in Section~\ref{sec-folklore}
and using Theorem~\ref{thm-folklore} 
we will work with perfect matchings of 
$\widehat{H}$.
Using Proposition~\ref{prop-rationalility}
we may assume that 
$\text{\emd}_{\!\!H}(L,R,\PP_L,\PP_R)=
\delta\times \frac{1}{q}$ for some integer $\delta>0$ and thus 
$
\text{\ric}_{\!\!G}(\ec)=
1- \text{\emd}_{\!\!H}(L,R,\PP_L,\PP_R)
= 1 -\frac{\delta}{q}
=
-\frac{\rho}{q}
$
for some integer $\rho=\delta-q\ge 0$. 
If $\rho=0$ then a trivially optimal solution involves the addition of just one suitable edge in $G$, thus we assume
that $\rho>0$.

Note that if $M$ is a perfect matching of 
$\widehat{H}$ 
then 
$\zeta(M)$ is an integer between $0$ and $3q$.
Let 
$\mcpm{{\widehat{H}}}$
be the cost of a minimum-cost perfect matching of $\widehat{H}$. 
For every $x\in
\{
\mcpm{{\widehat{H}}},
\mcpm{{\widehat{H}}}+1,
\dots,
3q
\}
$, 
we fix the following solutions and parameters for the rest of the proof:
\begin{enumerate}[label=$\triangleright$,leftmargin=*]
\item
For a perfect matching $M_x$ 
of
$\widehat{H}$ 
of 
cost $\zeta(M_x)=x$: 
\begin{enumerate}[label=$\triangleright$,leftmargin=*]
\item
For $i\in\{0,1,3\}$
let $n_{i,M_x}$ be the number of $i$-edges of $M_x$.
Let 
$n_{2,M_x}$ 
and 
$n_{2,M_x}'$ 
be the number of touchable and untouchable $2$-edges of $M_x$, respectively.
Thus, 
$q=n_{0,M_x}+n_{1,M_x}+n_{2,M_x}+n_{2,M_x}'+n_{3,M_x}=|\widehat{L}|=|\widehat{R}|$ 
and 
$\zeta(M_x)=n_{1,M_x}+2n_{2,M_x}+2n_{2,M_x}'+3n_{3,M_x}$.
Let $e_{1,M_x,i},\dots,e_{n_{i,M_x},M_x,i}$ be an arbitrary enumeration of the $i$-edges of $M_x$.
\item
Let $\widehat{\kappa}_{3,M_x},
\widehat{\kappa}_{2,M_x} \geq 0$ 
be 
defined as follows:
\begin{enumerate}[label=$\triangleright$,leftmargin=*]
\item
If $n_{3,M_x} \times \frac{2}{q} > 
\frac{\rho
+(x- \mcpm{{\widehat{H}}})
}{q} $ then 
$
\widehat{\kappa}_{2,M_x}=0$ 
and
$\widehat{\kappa}_{3,M_x}\leq n_{3,M_x}$ is the unique integer satisfying 
$(\widehat{\kappa}_{3,M_x}-1) \times \frac{2}{q} \leq \frac{\rho
+(x- \mcpm{{\widehat{H}}})
}{q}$
and 
$\widehat{\kappa}_{3,M_x} \times \frac{2}{q} > \frac{\rho
+(x- \mcpm{{\widehat{H}}})
}{q}$.
\item
Otherwise, 
if 
$
n_{3,M_x} \times \frac{2}{q} 
+ 
n_{2,M_x} \times \frac{1}{q} 
> \frac{\rho
+(x- \mcpm{{\widehat{H}}})
}{q} $ 
then 
$ \widehat{\kappa}_{3,M_x}=n_{3,M_x}$ 
and 
$\widehat{\kappa}_{2,M_x}\leq n_{2,M_x}$ is the unique integer satisfying 
$n_{3,M_x} \times \frac{2}{q} + (\widehat{\kappa}_{2,M_x}-1) \times \frac{1}{q} \leq \frac{\rho
+(x- \mcpm{{\widehat{H}}})
}{q}$
and 
$n_{3,M_x} \times \frac{2}{q} + \widehat{\kappa}_{2,M_x} \times \frac{1}{q} > \frac{\rho
+(x- \mcpm{{\widehat{H}}})
}{q}$.
\item
Otherwise, 
we call $M_x$ an ``unwanted'' matching and 
set 
$ 
\widehat{\kappa}_{3,M_x}=
\widehat{\kappa}_{2,M_x}=
\infty
$. 
\end{enumerate}
Let 
$\widehat{\kappa}_{M_x}=\widehat{\kappa}_{3,M_x}+\widehat{\kappa}_{2,M_x}$.
If 
$\widehat{\kappa}_{M_x}<\infty$
then 
we associate a subset $\cE_{M_x}$ of edges of $M_x$
with 
$\widehat{\kappa}_{M_x}$ in the following manner:
if 
$\widehat{\kappa}_{2,M_x}=0$
then
$
\cE_{M_x} = 
\big\{ 
e_{1,M_x,3},
\dots,
e_{\widehat{\kappa}_{3,M_x},M_x,3}
\big\}
$, 
otherwise
$
\cE_{M_x} = 
\big\{ 
e_{1,M_x,3},
\dots,
e_{n_{3,M_x},M_x,3},
\allowbreak
e_{1,M_x,2},
\dots,
e_{\widehat{\kappa}_{2,M_x},M_x,2}
\big\}
$
where 
$
e_{1,M_x,2},
\dots,
e_{\widehat{\kappa}_{2,M_x},M_x,2}
$
are touchable $2$-edges.
\end{enumerate}
\item
Let $\M_x$ denote a perfect matching 
in 
$\widehat{H}$ 
satisfying the following conditions:
\begin{itemize}
\item
$\M_x$ is \emph{not} an unwanted matching.
\item
$\zeta(\M_x)=x$. 
\item
$\M_x$
has a \emph{maximum} number of $3$-edges, say $\beta$ $3$-edges, among all perfect matchings of cost $x$ in 
$\widehat{H}$. 
\item
$\M_x$ has the maximum number of touchable $2$-edges
among all perfect matchings of $\widehat{H}$ of cost $x$ and having $\beta$ $3$-edges.
\end{itemize}
If there is no such matching then we set $\M_x=\emptyset$ and 
$\widehat{\kappa}_{\M_x}=\infty$.
\item
Let $\M$ be a perfect matching 
in 
$\widehat{H}$
such that $\M\eqdef \M_{\zeta(\M)}$
satisfies
$
\widehat{\kappa}_{\M}
=
\displaystyle
\min_{x}
\big\{ \widehat{\kappa}_{{\M}_x} \big\}
$.
Note that 
$\M$ is not an unwanted matching and 
$\widehat{\kappa}_{\M}<\infty$ 
since our problem has at least one valid solution.
\end{enumerate}

\begin{observation}\label{obs-repeat}
For any matching $M_x$ that is not unwanted the following claims hold:

\smallskip
\noindent
\textbf{(a)}
$\widehat{\kappa}_{M_x}$ 
is the minimum number of edges of $M_x$ whose weights may be decreased from their original values to $1$ 
to get a new perfect matching of weight strictly less than  
$\delta-\rho=q$.

\smallskip
\noindent
\textbf{(b)}
$
\widehat{\kappa}_{\M_x}
\geq
\widehat{\kappa}_{\M}
$
is true by definition of 
$\widehat{\kappa}_{\M}$.

\smallskip
\noindent
\textbf{(c)}
$
\widehat{\kappa}_{M_x}
\geq
\widehat{\kappa}_{\M_x}
$
is true for any matching $M_x$ of cost $x$ because,
by definition of 
$\M_x$,
$n_{3,\M_x} \geq 
n_{3,M_x} 
$
and 
if 
$
n_{3,\M_x}= n_{3,M_x}
$
then 
$
n_{2,\M_x}\geq n_{2,M_x}
$.
\end{observation}

For $\alpha\in\{0,1,2,3\}$ 
we define 
the $(i,j)\tx$ \emph{block} of $\alpha$-edges of $\widehat{H}$ as
$\cB^\alpha_{i,j}=
\big\{
\{ u_{i,k},v_{j,\ell} \} \,|\, k\in\{1,\dots,a\},\ell\in\{1,\dots,b\}, u_{i,k},v_{j,\ell}\in\widehat{H},
w(u_{i,k},v_{j,\ell})=\alpha
\big\}
$.
We say that the edges in 
$\cB^\alpha_{i,j}$
``correspond'' to the edge 
$\{u_i,v_j\}$ in $H$.
Abusing notations slightly, let 
$
\cB^{\alpha}_1,
\cB^{\alpha}_2,
\dots,
$
be an enumeration of 
the blocks of $\alpha$-edges in $\widehat{H}$.

\begin{lemma}\label{lemma-upper}
$
\opt(G,u,v)\leq\widehat{\kappa}_{\M}
$
and 
we can identify a set of at most 
$\widehat{\kappa}_{\M}$ edges in $G$ that is a valid solution.
\end{lemma}

\begin{proof}
Consider the set of 
$\widehat{\kappa}_{\M}$
edges in $\cE_{\M}$.
Suppose the edges in 
$\cE_{\M}$
appear in  
$
k \leq \widehat{\kappa}_{\M}
$
blocks of $\widehat{H}$, say the blocks 
$\cB_1^{\alpha_1},\dots,
\cB_k^{\alpha_k}
$
for $\alpha_1,\dots,\alpha_k\in\{2,3\}$
corresponding to the edges 
$
e_1^{\alpha_1},
\dots,
e_k^{\alpha_k}
$
of $H$. 
We select the set of $k$ edges 
$
E'=
\{
e_1^{\alpha_1},
\dots,
e_k^{\alpha_k}
\}
$
of $H$
as our solution (these are associated with a corresponding $k$ edges in $G$). 
Since $\cE_{\M}\subseteq \M$ and $\M$ is a wanted matching, 
decreasing the weights of the edges in $E'$ from $2$ or $3$ to $1$ 
results in a decrease of the cost of $\M$ in $\widehat{H}$, the new
cost of $\M$ being \emph{strictly} less than $\delta-\rho$, 
which in turn implies that 
the graph $H$ with these new weights satisfy
$
\text{\emd}_{\!\!H}(L,R,\PP_L,\PP_R) < \frac{\delta-\rho}{q}=1
$,
resulting in 
$\text{\ric}_{\!\!G}(\ec)=
1- \text{\emd}_{\!\!H}(L,R,\PP_L,\PP_R)>0$.
\end{proof}

\begin{lemma}~\label{lemma-lower}
$
\widehat{\kappa}_{\M} \leq b \times \opt(G,u,v)
$.
\end{lemma}

\begin{proof}
Let $\eopt(G,u,v)=\big\{e_1^{\alpha_1},\dots,e_{\opt(G,u,v)}^{\alpha_{\opt(G,u,v)}} \,|\, \alpha_1,\dots,\alpha_{\opt(G,u,v)}\in\{2,3\} \big\}$ 
be a set of $\opt(G,u,v)$ edges of $H$ in an optimal solution, 
and 
let
the edges 
$e_1^{\alpha_1},\dots,e_{\opt(G,u,v)}^{\alpha_{\opt(G,u,v)}}$ 
correspond
to the blocks 
$
\cB_1^{\alpha_1},
\dots,
\cB_{\opt(G,u,v)}^{\alpha_{\opt(G,u,v)}}
$ of $\widehat{H}$,
respectively.
Thus, decreasing the weights of 
$e_1^{\alpha_1},\dots,e_{\opt(G,u,v)}^{\alpha_{\opt(G,u,v)}}$ 
in $H$
correspond to decreasing the weights of edges in 
$
\widehat{\eopt}(G,u,v)=\cB_1^{\alpha_1} \,\cup\,
\dots
\,\cup\, \cB_{\opt(G,u,v)}^{\alpha_{\opt(G,u,v)}}
$ of $\widehat{H}$.
Let $\widehat{H}_{\downarrow }$ be the
graph obtained from $\widehat{H}$ by decreasing the weight of every edge in $\widehat{\eopt}(G,u,v)$ from its original 
value of $2$ or $3$ to $1$, 
and 
let
$M^\ast$
be a minimum-cost perfect matching of 
$\widehat{H}_{\downarrow }$
of cost 
$
\zeta(M^\ast)=
q\times\text{\emd}_{\!\! \widehat{H}_{\downarrow } } (L,R,\PP_L,\PP_R)
<
\delta-\rho=q
$.
Note that 
$
\cS=
M^\ast
\cap 
\widehat{\eopt}(G,u,v)
$
is a matching and the number 
$|\cS|$ 
of edges 
in $\cS$ 
is at most 
$\bbig \times \opt(G,u,v)$
since 
any matching for any of the 
$
\cB_j^{\alpha_j}
$'s
can have at most $\bbig$ edges.
Consider increasing the weights of the edges in $\cS$ from $1$ to their original 
values in $\widehat{H}$. 
This transforms $M^\ast$ to a perfect matching $M'$ of $\widehat{H}$ of 
cost 
$
\zeta(M')=x
$
for some  
$x\in
\{
\mcpm{{\widehat{H}}},
\mcpm{{\widehat{H}}}+1,
\dots,
3q
\}
$.
By using Observation~\ref{obs-repeat}
it follows that 
$
b \times \opt(G,u,v) \geq 
|\cS|
\geq
\widehat{\kappa}_{\M}
$.
\end{proof}

The only remaining part of the proof is to show 
how to compute $\M_x$.

\begin{lemma}
For each $x$, $\M_x$ can be computed by a polynomial time randomized algorithm with success probability $1-o(1)$.
\end{lemma}

\begin{proof}
For each $k,\ell\in\{0,1,\dots,q\}$ satisfying $k+\ell\leq q$ we 
compute a perfect matching $M_{k,\ell}$ of cost $x$ that has exactly 
$k$ $3$-edges and exactly $\ell$ touchable $2$-edges. Once all these 
matchings are computed, we can easily select $\M_x$ as one of them 
that satisfies our criteria. Thus, all that remains is to show 
how to compute each $M_{k,\ell}$.
Let 
$\widetilde{H}= (\widehat{L},\widehat{R},\widetilde{w})$
be the graph
obtained from 
$\widehat{H}$
by setting 
$\widetilde{w}(e)=4-\widehat{w}(e)$ 
for every edge $e$.
Then $M_{k,\ell}$ corresponds to a perfect matching (consisting of the same edges)
$\widetilde{M}_{k,\ell}$ 
in 
$\widetilde{H}$
of cost $\zeta(\widetilde{M}_{k,\ell})=4q-\zeta({M}_{k,\ell})=4q-x$
in which the number of $1$-edges is
exactly $k$
and the number of touchable $2$-edges is exactly $\ell$. 
Thus, equivalently, we can compute a perfect matching 
$\widetilde{M}_{k,\ell}$ 
of cost $4q-x$ with 
exactly $k$ $1$-edges and exactly $\ell$ touchable $2$-edges, 
and 
from this perfect matching we can reconstruct our desired perfect matching 
$M_{k,\ell}$.

We now show how to compute the perfect matching
$\widetilde{M}_{k,\ell}$
of cost 
$q_1=4q-x$. 
Noting that the sum of weights of all edges in 
$\widetilde{H}$
is at most $4q^2$, 
we select two integers $K_2\gg 4q^2$ and $K_1\gg 4q^2 K_2$, 
say 
$K_2=(4q+10)^{10}-2$ and 
$K_1=(4q+20)^{20}-1$, 
add $K_1$ to the weight of each $1$-edge 
and
add $K_2$ to the weight of each touchable $2$-edge
in 
$\widetilde{H}$, 
resulting in a new graph 
$\widetilde{H}'$.
Then, 
$\widetilde{M}_{k,\ell}$
corresponds to a perfect matching 
in $\widetilde{H}'$
of
cost 
$\Delta=K_1k+K_2\ell+(q_1-k-2\ell)$.
Conversely, consider a perfect matching of cost exactly $\Delta$ 
in $\widetilde{H}'$
and 
suppose this perfect 
matching
selected $k'$ $1$-edges and $\ell'$ touchable $2$-edges
resulting in a matching of cost 
$\Delta'=K_1k'+K_2\ell'+\alpha$ 
where 
$
3 (q-k'-\ell') \leq \alpha \leq 
4 (q-k'-\ell')
$
(since the $q-k'-\ell'$ edges which are not of weight $1$ or $2$ must be of weight $3$ or $4$).
We claim that $\Delta'=\Delta$ if and only if $k'=k$ and $\ell'=\ell$. Consider the 
following cases: 
\begin{enumerate}[label=$\triangleright$,leftmargin=*]
\item
Suppose that 
$k'=k+\delta_1$ and $\ell'=\ell+\delta_2$
where $\delta_1>0$ and $\delta_2 \geq 0$.
Then 
$
\Delta'-\Delta=K_1\delta_1+K_2\delta_2+\alpha- 
(q_1-k-2\ell)
>0
$
since $K_1\gg 16 q^4$.
\item
Suppose that 
$k'=k+\delta_1$ and $\ell'=\ell-\delta_2$
where $\delta_1>0$ and $\delta_2 \geq 0$.
Then 
$
\Delta'-\Delta=K_1\delta_1-K_2\delta_2+\alpha- (q_1-k-2\ell)
>
\frac{1}{2} K_1 >0
$
since $\delta,\delta_2\leq q$
and $K_1>q^9 K_2$.
\item
Suppose that 
$k'=k-\delta_1$ and $\ell'=\ell+\delta_2$
where $\delta_1>0$ and $\delta_2 \geq 0$.
Then 
$
\Delta'-\Delta=-K_1\delta_1+K_2\delta_2+\alpha- 
(q_1-k-2\ell)
<0
$ 
since $\delta,\delta_2\leq q$
and $K_1>q^9 K_2$.
\item
Suppose that 
$k'=k-\delta_1$ and $\ell'=\ell-\delta_2$
where $\delta_1>0$ and $\delta_2 \geq 0$.
Then 
$
\Delta'-\Delta=-K_1\delta_1-K_2\delta_2+\alpha- 
<0
$
since $K_1\gg 16 q^4$.
\end{enumerate}
Thus it follows that $k'=k$ and now a similar argument 
involving 
$\Delta - K_1k$ (instead of $\Delta$)
shows that $\ell'=\ell$.
Now we use the 
randomized polynomial-time 
algorithm (with success probability of $1-o(1)$) for the 
EWPM problem
in~\cite{elma23}
(which itself uses results in~\cite{GAMST16,MVV87})
to get 
a matching of cost exactly $\Delta$ in 
$\widetilde{H}'$.
\end{proof}

\smallskip
\noindent
\textbf{(\emph{b})}
Consider a minimum-cost perfect matching $\M$ 
of $\widehat{H}$
of cost 
$\mcpm{{\widehat{H}}}=
\text{\emd}_{\!\!H}(L,R,\PP_L,\PP_R)
=
\frac{\delta}{q}
$, 
giving 
$
\text{\ric}_{\!\!G}(\ec)
= 1 -\frac{\delta}{q}
=
-\frac{\rho}{q}
$
for some integer $\rho=\delta-q\ge 0$. 
We modify the matching via 
Proposition~\ref{prop-struct-H}
so that it satisfies the criterion of 
Proposition~\ref{prop-struct-H}.
Note that 
if we decrease the weight of every $3$-edge and every touchable $2$-edge in $\M$, 
this is equivalent to adding every possible edge that we can insert in $G$ and, 
since our problem has a valid solution, this must result in a perfect matching of $\widehat{H}$ 
of cost strictly less than $q$ (since we assume that there is a valid solution).
Bases on the above observation, 
we execute the following simple algorithm: 
\begin{enumerate*}[label={\textbf{({\arabic*})}}]
\item\label{ttt} 
Repeatedly select an edge, say 
$\{u_{j_1,\ell_1},v_{j_2,\ell_2}\}$,  
in $\M$ 
with $w(u_j,v_j)=3$,
$w(u_{j_1,\ell_1},v_{j_2,\ell_2})=3$,  
decrease the weight of the edge 
$\{u_{j_1,\ell_1},v_{j_2,\ell_2}\}$
to $1$, and stop if the 
total cost of the perfect matching  
with the modified weight is
strictly less than $q$.
\item\label{ttt2} 
If there are no more such $3$-edges in $\M$ as mentioned in~\ref{ttt}
and total cost of the perfect matching  
has not still dropped below $q$
then
repeatedly select a touchable $2$-edge matched by $\M$, say 
$\{u_{j_1,\ell_1},v_{j_2,\ell_2}\}$,  
with $w(u_{j_1,\ell_1},v_{j_2,\ell_2})=2$,  
decrease the weight of the edge 
$\{u_{j_1,\ell_1},v_{j_2,\ell_2}\}$
to $1$, and stop 
if the 
total cost of the perfect matching  
with the modified weight is
strictly less than $q$.
\end{enumerate*}
Let $p$ be the number of edges whose weights were decreased by our algorithm and 
let $\mu$ be the final value of the \emd\ of $H$ at the end of our algorithm. 
If our algorithm terminated in Step~\ref{ttt} 
then 
we have 
$\mu=1-\frac{2}{q}=\frac{q-2}{q}$
and 
$p=
\left\lceil\frac{\delta - (q-2)}{2}\right\rceil
=
\left\lceil\frac{\rho+2}{2}\right\rceil
\leq
\rho+1
$.
If our algorithm terminated in Step~\ref{ttt2} 
then 
we have 
$\mu=1-\frac{1}{q}=\frac{q-1}{q}$
and 
$p\leq \delta - (q-1)
=
\rho+1
$.
Thus, in either case, 
$p\leq \rho+1$.
The proof of Lemma~\ref{lemma-upper}
shows that we can
select a set of at most $p$ edges in 
$H$
that form a valid solution of our problem.

Conversely, 
the proof of Lemma~\ref{lemma-lower}
shows that 
there exists a perfect matching $M''$ of $\widehat{H}$ of cost
$
\zeta(M'')=x\geq 
\delta
$
such that 
decreasing some set $\cS$ of weights of some set of touchable $2$-edges and $3$-edges of 
$M''$ results in a 
perfect matching $M'$ of cost 
$\zeta(M')<\delta-\rho$
and $\opt(G,u,v)\geq \frac{|\cS|}{\bbig}$.
We now repeat essentially a similar analysis in the previous paragraph for 
$M''$.
If $\cS$ consists of only $2$-edges then 
$|\cS|=
\left\lceil\frac{x - (q-2)}{2}\right\rceil
\geq
\left\lceil\frac{\delta - (q-2)}{2}\right\rceil
=
\left\lceil\frac{\rho+2}{2}\right\rceil
\geq
\frac{\rho+1}{2}
$.
If $\cS$ contained at least one $1$-edge then 
$|\cS|\geq 
\frac{x - (q-1)}{2}
\geq
\frac{\delta - (q-1)}{2}
=
\frac{\rho+1}{2}
$.
Thus it follows that 
$\opt(G,u,v)\geq \frac{\rho+1}{2\bbig}$.

For $\bbig=1$ the following example 
shows that the performance analysis of our algorithm 
is tight (see \FI{fig3-cap} for a pictorial illustration when $m=8$). 
Let $\deg_G(u)=\deg_G(v)=m+1$ for an even integer $m>2$
and assume that $\nbr_G(u)\cap\nbr_G(v)=\emptyset$.
Note that the two graphs $H$ and $\widehat{H}$
are the same when $\bbig=1$ and thus we can directly work with $H$. 
Set the edge weights of $H$ as follows:
$w(u_i,v_j)=1$ for $(i,j)=(1,1),(i,j)=(2,2),\dots,(i,j)=(\frac{m}{2},\frac{m}{2})$,
$w(u_i,v_j)=2$ for 
$(i,j)=(1,1+\frac{m}{2}),(i,j)=(2,2+\frac{m}{2}),\dots,(i,j)=(\frac{m}{2},m)$
and for 
$(i,j)=(1+\frac{m}{2},1),(i,j)=(2+\frac{m}{2},2),\dots,(i,j)=(m,\frac{m}{2})$,
and 
$w(u_i,v_j)=3$ for 
all other edges.
Suppose that we start with the minimum-cost perfect matching 
$M=
\big\{
\{u_i,v_j\} \,|\, i\in\{1,\dots,\frac{m}{2}\}, j=i+\frac{m}{2}
\big\}
\,\cup\,
\big\{
\{u_i,v_j\} \,|\, j\in\{1,\dots,\frac{m}{2}\}, i=j+\frac{m}{2}
\big\}
$
of cost $\zeta(M)=2m$.
Then, we need decrease the weight of every edge in $M$ to $1$, 
thus using $m$ edges.
However, the optimal solution may instead start with 
the minimum-cost perfect matching 
$
M'=
\big\{
\{u_i,v_i\} \,|\, i\in\{1,\dots,m\}
\big\}
$
and need to decrease only the weights of the $3$-edges to $1$, thereby
using $\frac{m}{2}$ edges.
The example can be generalized for other values of $\bbig$.
\hfill{\Pisymbol{pzd}{113}}


\begin{figure}[tb]
\scalebox{1}[1]{\includegraphics{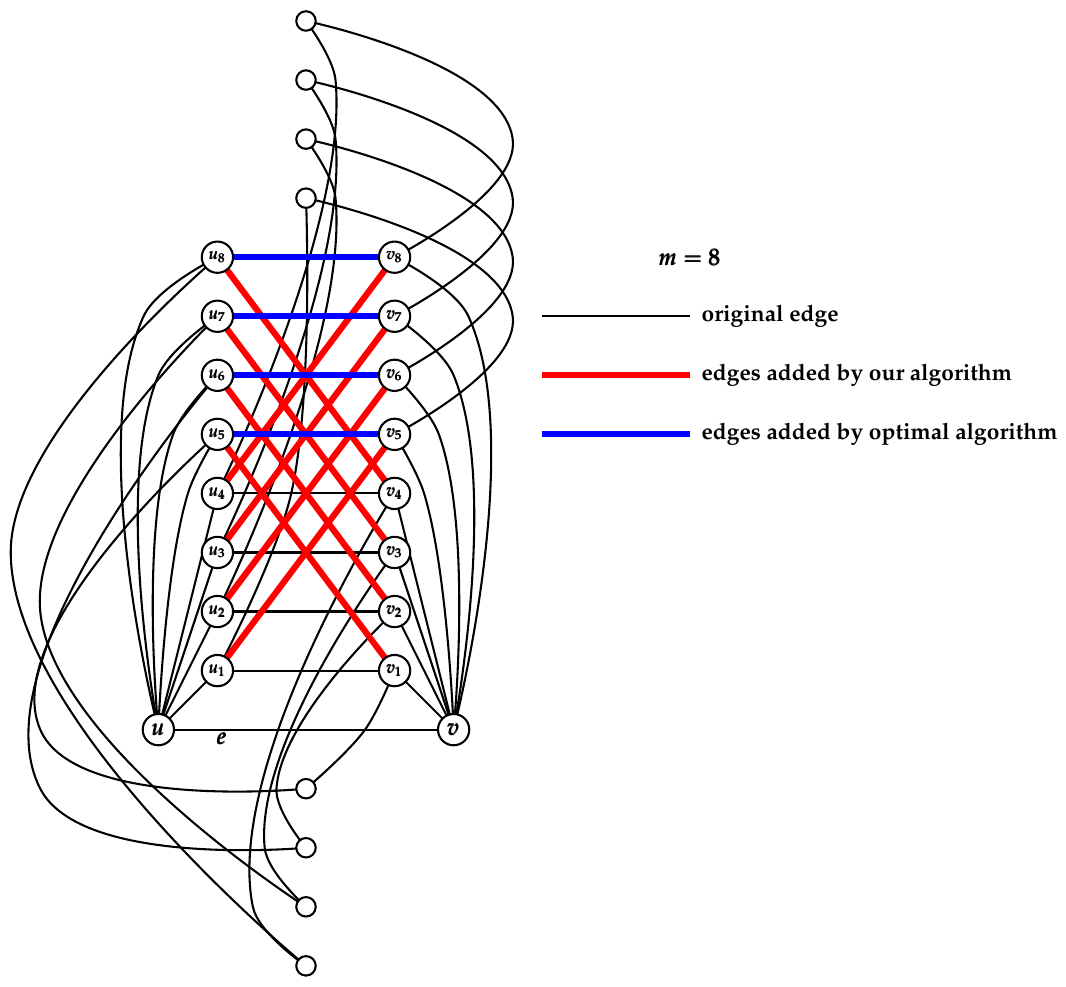}}
\caption{\label{fig3-cap}
Pictorial illustration (for $m=8$) of the example stated in the proof of 
Theorem~\ref{thm-alg-ins-rest}\textbf{\emph{(b)}}
showing the tightness of the algorithm presented there for the case when 
$\bbig=1$.
}
\end{figure}

\subsubsection{Polynomial Time Algorithms for Unrestricted (General) Instances of Problem~\uw-\rt-\ins-\ntp}
\label{sec-uw-rt-ins-ntp-unrest}

\begin{theorem}\label{thm-alg-ins-unrest}
Consider a general instance of 
{\em Problem}~{\rm\uw-\rt-\ins-\ntp}
which need not satisfy property 
\textbf{$(\spadesuit)$} 
and 
suppose that it  
has at least one valid solution.
Then the following claims are true.
\begin{description}
\item[\textbf{\emph{(a)}}]
There is a randomized polynomial time algorithm that returns a valid solution $E'$ of 
{\em Problem}~\uw-\rt-\ins-\ntp\
with a success probability of $1-o(1)$
satisfying 
$|E'|\leq (\abig+\bbig)\times \opt(G,u,v)$. 
\item[\textbf{\emph{(b)}}]
There is a deterministic polynomial time algorithm that returns a valid solution $E'$ of 
{\em Problem}~\uw-\rt-\ins-\ntp\
satisfying 
$|E'|\leq 2(\abig+\bbig)\times \opt(G,u,v)$.
\end{description}
\end{theorem}

\begin{corollary}
If 
$\deg_G(v)=\deg_G(u\,\equiv\,\abig=\bbig=1$ then 
{\em Problem}~\uw-\rt-\ins-\ntp\
admits a randomized $2$-approximation
and 
a deterministic $4$-approximation.
\end{corollary}

\noindent
\textbf{Proof of Theorem~\ref{thm-alg-ins-unrest}}
{\em We reuse the notations, observations and the proofs presented in Section}~\ref{sec-folklore}.
We use the same proof approach as in 
Theorem~\ref{thm-alg-ins-rest},
but we need some new results and modification of some bounds in 
Theorem~\ref{thm-alg-ins-rest}.
Since there may be one or more edges connecting pairs of nodes in $\nbr_G(v)\setminus\{u\}$ or pairs of nodes in 
$\nbr_G(u)\setminus\{v\}$, insertion of an $\{u_i,v_j\}$ edge in $G$ may correspond to change of 
more than one edge in $H$ and $\widehat{H}$.
Thus, we first need to characterize the nature of the effect of change of weight of one edge in $H$
on other edges in $H$ and their corresponding effects in $\widehat{H}$.

\begin{observation}[mega change of edge weights]
Consider the graphs $G$, $H$ and $\widehat{H}$
and suppose that we have inserted some new edges in $G$ 
resulting in the graphs 
$G'$, $H'=(L,R,w')$ and $\widehat{H'}=(\widehat{L},\widehat{R},\widehat{w'})$. 
Suppose now we change the weight of an edge 
$\{u_i,v_j\}$ in $H'$ from $3$ or $2$ to $1$ corresponding to insertion of the edge 
$\{u_i,v_j\}$
in $G'$.
Then the following observations hold and definitions apply.
\begin{description}[leftmargin=10pt]
\item[($\star$)]
Excluding the edge $\{u_i,v_j\}$ itself, 
the only edges in $H'$ 
whose weights may decrease are of two forms: 
\begin{adjustwidth}{0.6cm}{}
\begin{description}
\item[(I)]
an edge 
$\{u_x,v_j\}$ 
for some $u_x\in L$
corresponding to an edge 
$\{u_i,u_x\}$ in $G$,
or 
\item[(II)]
an edge
$\{u_i,v_y\}$ 
for some $v_y\in R$
corresponding to an edge 
$\{v_j,v_y\}$ 
in $G$.
\end{description}
\end{adjustwidth}
In particular such an edge-weight change occurs only when 
$w'(u_x,v_j)=3$ 
in \textbf{(I)}
or
$w'(u_i,v_y)=3$
in \textbf{(II)},
and the weight in that case is changed from $3$ to $2$.
\item[($\star\star$)]
The change in the weight of the edge 
$\{u_x,v_j\}$ 
\emph{(\emph{respectively, edge 
$\{u_i,v_y\}$})} 
in 
\textbf{($\star$)}
lead to a corresponding change, from $3$ to $2$,
of all the edges in the blocks 
$\cB_{x,j}\eqdef\bigcup_{\alpha\{0,1,2,3\}} \cB^\alpha_{x,j}$
\emph{(\emph{respectively, blocks 
$\cB_{i,y}\eqdef\bigcup_{\alpha\{0,1,2,3\}} \cB^\alpha_{i,y}$})}.
We will use the following notations:  
\begin{gather*}
\textstyle
\cC_{i,j}^{(j)} \eqdef 
\bigcup_{
\text{\small$\{u_i,u_x\}$ in $G$}
}  \cB_{x,j}, 
\,\,\,\,
\cC_{i,j}^{(i)} \eqdef 
\bigcup_{
\text{\small$\{v_j,v_y\}$ in $G$}
}  \cB_{i,y}, 
\,\,\,\,
\cC_{i,j}^{(i,j)} \eqdef 
\cC_{i,j}^{(i)} \cup 
\cC_{i,j}^{(j)}
\end{gather*}
\end{description}
\end{observation}

\smallskip
\noindent
\textbf{(\emph{a})}
We use the same proof approach as in 
Theorem~\ref{thm-alg-ins-rest}(\emph{a}) but with a modified version of 
Lemma~\ref{lemma-lower}.

\begin{lemma}[\emph{cf}. Lemma~\ref{lemma-lower}]~\label{lemma-lower2}
$
\widehat{\kappa}_{\M} \leq (a+b) \times \opt(G,u,v)
$.
\end{lemma}

\begin{proof}
We use the proof approach as in 
Lemma~\ref{lemma-lower}
and 
point out the modifications needed.
Let 
$e_j^{\alpha_j}
$
be the edge
$ \{ u_{\ell_j}, v_{p_j} \} $
in $H$ where 
$
\ell_j\in\{1,\dots,r\},
p_j\in\{1,\dots,s\},
w( u_{\ell_j}, v_{p_j}) =\alpha_j
$.
In the proof of 
Lemma~\ref{lemma-lower}
we argued that 
increasing the weight of the edge 
$ \{ u_{\ell_j}, v_{p_j} \} $
can only change the weights of at most $\bbig$ edges corresponding to a 
matching in the block 
$
\cB_j^{\alpha_j}
$.
In our present case, however, 
increasing the weight of the edge 
$ \{ u_{\ell_j}, v_{p_j} \} $
can change the weights 
of a matching of 
$\cC_{\ell_j,p_j}^{(\ell_j,p_j)} \cup \cB_j^{\alpha_j}$
which is a union 
of a matching 
in 
$\cC_{\ell_j,p_j}^{(\ell_j)}$,
a matching in 
$\cC_{\ell_j,p_j  i,j}^{(p_j)}$,
and 
a matching in 
$
\cB_j^{\alpha_j}
$.
The maximum number of edges in a matching 
in 
$\cC_{\ell_j,p_j}^{(\ell_j)}$
(\emph{resp}., 
$\cC_{\ell_j,p_j}^{(p_j)}$)
is $\abig$ 
(\emph{resp}., $\bbig$) since one end-point of such a matching must be 
a node from 
$
\{ u_{\ell_j,k} \,|\, k\in\{1,\dots,a\}
$
(\emph{resp}., 
a node from 
$
\{ v_{p_j,y} \,|\, y\in\{1,\dots,b\}
$),
\IE, in the notations used in 
the proof of Lemma~\ref{lemma-lower}
we have 
$
\widehat{\eopt}(G,u,v) \subseteq 
\bigcup_{j=1}^{\opt(G,u,v)} 
\big ( \cB_j^{\alpha_j} \cup
\cC_{\ell_j,p_j}^{(\ell_j,p_j)} \big )
$
and 
$
|\cS|=
|M^\ast \cap \widehat{\eopt}(G,u,v) |
\leq (a+b)\times 
\opt(G,u,v)
$.
Thus, 
increasing the weight of the edge 
$ \{ u_{\ell_j}, v_{p_j} \} $
can change the weights of at most $\abig+\bbig$ edges.
With this modification, a proof similar to 
the proof of 
Lemma~\ref{lemma-lower}
proves the bound that 
$
\widehat{\kappa}_{\M} \leq (a+b) \times \opt(G,u,v)
$.
\end{proof}

\smallskip
\noindent
\textbf{(\emph{b})}
The same algorithm as in 
Theorem~\ref{thm-alg-ins-rest}(\emph{b})
works after noting that 
decreasing the weight of the edge 
$ \{ u_{\ell_j}, v_{p_j} \} $
can change the weights of at most $\abig+\bbig$ edges.
\hfill{\Pisymbol{pzd}{113}}

\subsection{$\NP$-completeness of \uw-\rt-\del-\ptn}
\label{sec-uw-rt-del-ptn}

For reader's convenience, we write down 
Problem~\wt-\rt-\del-\ntp\
precisely as follows. 
The input to the problem is an 
undirected unweighted graph $G=(V,E)$ 
and
an edge $\ec=\{u,v\}\in E$ such that $\text{\ric}_{\!\!G} (\ec)>0$.
A valid solution of the problem is a graph 
$G^{-E'}$ for a subset of edges $E'\subset E$ 
satisfying $\{u',v'\}\in E' \Rightarrow u',v'\notin \{u,v\}$
such that $\text{\ric}_{\!\!G^{-E'}} (\ec)<0$,
and the 
objective is to \emph{minimize} $|E'|$.
We use the notation $\eopt(G,u,v)$ to denote an optimal solution of the problem 
of value $\opt(G,u,v)=|\eopt(G,u,v)|$.

\begin{theorem}\label{thm-lower}~\\
\textbf{\emph{(\emph{a})}}
It is possible to decide in polynomial time if 
{\em Problem}~\uw-\rt-\del-\ptn\ 
has a valid solution.
Moreover, 
if $\deg_G(v)>\frac{3}{2}\eta+5$ where 
$\eta$ is the number of common neighbors of $u$ and $v$ 
then there is always a valid solution.

\smallskip
\noindent
\textbf{\emph{(\emph{b})}}
{\em Problem}~\uw-\rt-\del-\ptn\ is {\em NP}-complete
even under the constraint $\deg_G(u)=\deg_G(v)$.
\end{theorem}

\begin{proof}
{\em We reuse the notations, observations and the proofs presented in Section}~\ref{sec-folklore}.

\medskip
\noindent
\textbf{(\emph{a})}
We delete every possible edge that we are allowed to delete resulting in 
a graph $G'$ and check if 
$\text{\ric}_{\!\!G'}(\ec)<0$;
if not, the problem has no valid solution.
Let $\eta\geq 0$ be the number of common neighbors of $u$ and $v$.
For $j\in \{0,1,\dots,\eta+1 \}$,
by Proposition~\ref{prop-struct-H}
in $G'$
demand for the node $v_{s-j}$ is already satisfied at zero transportation cost 
and the node $u_{s-j}$ has an amount of 
$\frac{1}{r} - \frac{1}{s} $
left to transport to other nodes. 
Note that
after edge removals 
$
\dist_{G'}( u_{s-j}, v_{s-j'} ) =3
$
for $j,j'>\eta+1$
and 
$
\dist_{G'}( u_{s-j}, v_{s-j'} ) =2
$
for $j>1,j'>\eta+1$.
We get a lower bound for 
$\text{\emd}_{\!\!G'}(u,v)$
from the following considerations.
The nodes 
$v_1,\dots, v_{s-(\eta+2)}$
has a total demand of 
$
\frac{s-\eta-2}{s}
$.
Out of this amount, they may receive 
an amount of 
$\frac{1}{r} - \frac{1}{s} $ by paths of length $1$ (from $v_s$),
an amount of 
$\frac{1}{r} - \frac{1}{s} $ by paths of length $2$ (from $v_{s-1}$),
and the remaining amount 
$
\frac{s-\eta-2}{s}
-
\frac{2}{r} + \frac{2}{s}
$
must be received by paths of length $3$.
Thus, 
$\text{\emd}_{\!\!G}(u,v)
\geq 
\frac{3}{r} - \frac{3}{s}
+
\frac{3s-3\eta-6}{s}
-
\frac{6}{r} + \frac{6}{s}
=
\frac{3s-3\eta-3}{s}
-
\frac{3}{r} 
$.
Therefore 
$\text{\emd}_{\!\!G}(u,v)>1$ 
can be obtained by setting 
$
\frac{3s-3\eta-3}{s}
-
\frac{3}{r} 
>1
$.
Noting that $\eta+2\leq r\leq s$,
the above inequality is satisfied 
if 
$
\frac{3s-3\eta-3}{s}
-
\frac{3}{\eta+2} 
>1
$,
which in turn is satisfied provided 
$s>\frac{3}{2}\eta+6$.

\medskip
\noindent
\textbf{(\emph{b})}
We reduce from 
the minimum blocker perfect matching problem (\textsc{Mbpmb}) 
problem~\cite{LACROIX201225} is as follows:
{\em 
given a 
bipartite graph $H_0=(L_0,R_0,E_0)$ with $|L_0|=|R_0|=n$ which has at least one perfect matching,  
find a 
subset of edges $B\subseteq E_0$ of \emph{minimum} cardinality such that 
$H_0^{-B}$
does not admit a perfect matching}.
\textsc{Mbpmb}
was shown to be NP-complete in~\cite{LACROIX201225}.
We consider such a bipartite graph $H_0$ with 
$
L_0=\{u_1,\dots,u_n\},
R_0=\{v_1,\dots,v_n\},
$
and 
construct the graph $H=(L,R,w)$ of 
Section~\ref{sec-H-to-anomaly} in the following manner:
$
L=L_0\cup \{ u_{n+1}, u_{n+2}=u_{r-1}=u, u_{n+3}=u_r=v \},
R=R_0\cup \{ v_{n+1}, v_{n+2}=v_{r-1}=u, v_{n+3}=v_r=v \}
$,
$
w(u_i,v_j)=1
$
for $i,j\in\{1,\dots,n\}$ and $\{u_i,v_j\}\in E_0$,
$
w(u_i,v_j)=3
$
for $i,j\in\{1,\dots,n\}$ and $\{u_i,v_j\}\notin E_0$,
and 
$
w(u_{n+1},v_{n+1})=2
$, 
where $u_{n+1}$ and $v_{n+1}$ are two new nodes. 
In the original graph $G=(V,E)$, 
$\nbr_G(u)=L\setminus \{u\}$, 
$\nbr_G(v)=R\setminus \{v\}$, 
$\deg_G(u)=
\deg_G(v)=
|L|-1 =|R|-1 = n+2$, 
and 
$w(u_{n+1},v_{n+1})=2$ is implemented by adding a new node $x$ and adding the two edges 
$\{u_{n+1},x\},
\{v_{n+1},x\}$ (the node $x$ is not connected to any other node in $G$).
Since $|L|=|R|$ implies $\bbig=1$, the two graphs $H$ and $\widehat{H}$ are the same,
using Proposition~\ref{prop-struct-H}
a minimum-cost perfect matching of $H$ consists of a perfect matching of 
$H_0$ (consisting of $n$ $1$-edges), two edges of weight $0$ and one edge of weight $2$.
Thus, using Theorem~\ref{thm-folklore} we get 
$
\text{\ric}_{\!\!G}(u,v)
=
1- \text{\emd}_{\!\!G}(u,v)
=
1-\text{\emd}_{\!\!H}(L,R,\PP_L,\PP_R)
=
1-\frac{n+2}{n+3}
=\frac{1}{n+3}>0
$.

Consider a solution $B\subseteq E_0$ of  
\textsc{Mbpmb}.
Then, after changing the weights of the edges in $B$ from $1$ to $3$ in $H$, $H$ does not contain a perfect 
matching of $1$-edges and so a minimum-cost perfect matching of $H$ must contain at least one $3$-edge, 
giving 
$
\text{\ric}_{\!\!G}(u,v)
=
1- \text{\emd}_{\!\!G}(u,v)
=
1-\text{\emd}_{\!\!H}(L,R,\PP_L,\PP_R)
=
1-\frac{n+4}{n+3}
=-\frac{1}{n+3}<0
$.
Conversely, suppose that $E'$ is a valid solution of 
\uw-\rt-\del-\ptn. 
Then, $H^{-E'}$ cannot contain a perfect matching of the $1$-edges and thus 
$E'$ is a valid solution of 
\textsc{Mbpmb}.
\end{proof}

\section{Unrestricted Edge Insertions for Weighted Graphs
(\wt-\ut-\ins-\ntp)}
\label{wt-ut-ins-ntp}

For reader's convenience, we write down 
Problem~\wt-\ut-\ins-\ntp\
precisely as follows. 
The input to the problem is an 
undirected weighted graph $G=(V,E,w)$ 
with each edge weight being an integer from the set $\{1,2,\dots,W\}$,
and
an edge $\ec=\{u,v\}\in E$ such that $\text{\ric}_{\!\!G} (\ec)<0$.
A valid solution of the problem is a graph $G^{+E'}$ for a set of node pairs 
$E'=\{\{x,y\}\,|\, x\neq y,\{x,y\}\notin E \}$
such that $\text{\ric}_{\!\!G^{+E'}} (\ec)>0$, and the 
objective is to \emph{minimize} $|E'|$.
We use the notation $\eopt(G,u,v)$ to denote an optimal solution of the problem 
of value $\opt(G,u,v)=|\eopt(G,u,v)|$.

\begin{theorem}\label{thm-hard-ins}$\!\!$\footnote{It is not uncommon that some inapproximability results in 
computational complexity holds up to a specific value of a constant $\eps>0$ as opposed to all values of $\eps$;
for example, it is hard to approximate 
the unique coverage problem within a factor of $\Omega(1/\log^\sigma n)$, for some constant $\sigma$ depending
on $\eps$, assuming that $\NP\not\subseteq\mathrm{BPTIME}(2^{n^\eps})$ for some $\eps>0$~\cite{DFHS08}.}
Assuming \emph{P}$\neq$\emph{NP},
for some constant $\eps>0$ 
Problem \wt-\ut-\ins-\ntp\ cannot be approximated within a factor of 
$\eps\ln(\deg_G(v))$ 
in polynomial time
even if $\bbig=\frac{\mathrm{LCM}(\deg_G(u)+1,\deg_G(v)+1)  }{\deg_G(v)+1 }=1$.
\end{theorem}

\begin{proof}
The standard (unweighted) minimum set cover (\msc)
problem is defined as follows.
We have a universe $\cU=\{a_1,\dots,a_n\}$ consisting of $n$ elements and 
a collection $\cT$ of $m$ sets $\cS_1,\dots,\cS_m\subseteq\cU$ with $\cup_{j=1}^m \cS_j=\cU$.
A valid solution is a collection of sets whose union is $\cU$, and 
the objective is to \emph{minimize} the number of sets in a valid solution.
Let 
$\opt_{\!\!\text{\footnotesize\msc}}$
denote
the number of sets in an optimal solution of \msc. 
Research works of 
Feige~\cite{Feige98}
and 
Dinur~\cite{Dinur14}
provide a 
logarithmic 
inapproximability of \msc\ assuming P$\neq$NP and imply the following result.

\begin{theorem}{\rm~\hspace*{-0.1in}\cite{Feige98,Dinur14}}\label{bk}
For any language $L$ in $\NP$,
there exists a polynomial-time reduction
such that,
for any constant $0<\eps<1$, 
given an instance
$I$ of $L$ produces 
a number $\xi$ 
and 
an instance of 
$\langle\cU_1,\cT_1\rangle$ 
of {\rm \msc} with $n_1$ elements and $m_1\leq {n_1}^{O(1)}$ sets
such that:
%
\begin{itemize}
\item
$n_1$ is an even number;
\item
if $I\in L$ then 
$\opt_{\!\!\text{\footnotesize\msc}}\leq\xi$; 
\item
if $I\not\in L$ then every 
valid solution of the \msc\ instance requires at least 
$(\eps\ln n_1) \xi$ sets (and thus 
$\opt_{\!\!\text{\footnotesize\msc}}\geq (\eps\ln n_1) \xi$). 
\end{itemize}
%
\end{theorem}

Consider the instance of \msc\ as generated by Theorem~\ref{bk}. Let $m_1\leq {n_1}^c$ for some integer constant $c\geq 1$ 
for that instance and let $N={n_1}^{10c}$.
For each element $a_i$ and each set $\cS_j$ such that 
$a_i\in\cS_j$ we replace $a_i$ in $\cS_j$ by $N$ elements 
$a_{i,1},\dots,a_{i,N}$. 
Clearly this does not change the value of 
$\opt_{\!\!\text{\footnotesize\msc}}$ 
or the number of sets $m'$ 
but this does change the universe $\cU_1$ to a bigger universe $\cU$ with $n=|\cU|={n_1}^{10c+1}$. 
Since $\eps\ln n_1 = \frac{\eps}{10c+1}\ln n$, this 
gives the following result:
\begin{quote}
$(\star)$ for any language $L$ in $\NP$,
there exists a polynomial-time reduction
such that,
\emph{for some universal constant} $0<\eps<1$, 
given an instance
$I$ of $L$ produces 
a number $\xi$ 
and 
an instance of 
$\langle\cU,\cT\rangle$ 
of {\rm \msc} with an even number $n$ of elements and $m<n^{0.1}$ sets
such that
if $I\in L$ then 
$\opt_{\!\!\text{\footnotesize\msc}}\leq\xi$ but 
if $I\not\in L$ then every 
valid solution of the \msc\ instance requires at least 
$(\eps\ln (n+1)) \xi$ sets (and thus 
$\opt_{\!\!\text{\footnotesize\msc}}\geq (\eps\ln (n+1)) \xi$). 
\end{quote}
%
We call an instance of \msc\ as generated in $(\star)$ as a ``hard'' instance. 
Given such a hard instance of \msc, we generate an instance 
$G=(V,E,w)$ 
of 
Problem \wt-\ut-\ins-\ntp\ in the following manner: 
\begin{enumerate}[label=$\triangleright$,leftmargin=*]
\item
Node $u$ has no neighbors other than $v$, \IE, $\deg_G(u)=1$, and $w(u,v)=2$.
\item
Other than $u$, node $v$ has $n={n_1}^{10c+1}$ neighbors corresponding to the elements of $\cU$; 
we call them ``element'' nodes.
Let $v_{a_{i,j}}$ be the neighbor corresponding to the element 
$a_{i,j}\in\cU$. The weight of each edge 
$\{v,v_{a_{i,j}}\}$ is set to be $W={n_1}^{100c}$; we call each such edge as a ``heavy'' edge.
Note that $\deg_G(v)+1=n+2$ and thus
$\bbig=\frac{\mathrm{LCM}(\deg_G(u)+1,\deg_G(v)+1)  }{\deg_G(v)+1 }=1$.
\item
For every set $\cS_k$ we have a ``set'' node 
$v_{\cS_k}$ with the $|\cS_k|$ ``membership'' edges 
$
\{
\{ v_{\cS_k}, v_{a_{i,j}} \} \,|\, a_{i,j}\in\cS_k
\}
$, each of
weight $1$.
\end{enumerate}
We will reuse the notations and proofs in 
Sections~\ref{sec-emd}--\ref{sec-folklore}
as needed.
Note that since $w(u,v)=2$ by 
Observation~\ref{obs-inswt}(\emph{b})
every inserted edge has a weight of $1$.

\begin{proposition}
$\text{\emd}_{\!\!G}(u,v) >2$
and thus 
$\text{\ric}_{\!\!G}(u,v) <0$.
\end{proposition}

\begin{proof}
The total demand for all the element nodes is 
$
\frac{n}{n+2}
$
and the only way to reach them is via a heavy edge. 
Thus, 
$\text{\emd}_{\!\!G}(u,v) >
\frac{nW}{n+2}
=
\frac{{n_1}^{10c+1}\times {n_1}^{100c} }{{n_1}^{10c+1}+2}
>2
$
and hence 
$\text{\ric}_{\!\!G}(u,v) = 1 - \frac{ \text{\emd}_{\!\!G}(u,v) } { \dist_G(u,v) } <0$.
\end{proof}

We observe that $\dist_G(u,v)$ remains $2$ even if every possible edge was inserted in $G$.
We complete our proof by showing the following: 

\begin{adjustwidth}{0.6cm}{}
\begin{description}
\item[(completeness)]
If 
$\opt_{\!\!\text{\footnotesize\msc}}\leq\xi$
then 
$\opt(G,u,v)\leq\xi$.
\item[(soundness)]
If 
$\opt_{\!\!\text{\footnotesize\msc}}\geq (\eps\ln (n+1)) \xi$
then 
$\opt(G,u,v) \geq (\eps\ln \deg_G(v)) \xi$.
\end{description}
\end{adjustwidth}

\medskip
\noindent
\textbf{Proof of completeness}

\smallskip
For notational convenience, let $m'=\opt_{\!\!\text{\footnotesize\msc}}\leq\xi\leq m <n^{0.1}$.
Let $\cS_1,\dots,\cS_{m'}$ be an optimal solution of the \msc.
We insert the edges  
$\{u,v_{\cS_1}\},\dots,\{u,v_{\cS_{m'}}\}$, each of weight $1$.
Let $G'$ be the new graph.
Note that the shortest path between any two nodes in $G'$ avoid the heavy edges.
Note that 
$
1-\frac{2}{m'+2}
<
1-\frac{2}{n+2}
$
since $m'<n$.
The heavy nodes have a total demand of $\frac{n}{n+2}=1-\frac{2}{n+2}$ of dirts. 
Out of that, they can receive a total amount of 
$
\frac{m'}{m'+2}=1-\frac{2}{m'+2}
$ 
from the 
set nodes via edges of weights $1$, resulting in a total transportation cost of 
$\frac{m'}{m'+2}=1-\frac{2}{m'+2}$.
All of the remaining shipments to all other nodes, 
totalling an amount of 
$
1-\frac{m'}{m'+2}=\frac{2}{m'+2}
$,
can be done via paths of length at most $4$. 
Thus, 
$
\text{\emd}_{\!\!G}(u,v) 
<
1-\frac{2}{m'+2}
+
\frac{8}{m'+2}
=1+
\frac{6}{m'+2}
$
and 
$\text{\ric}_{\!\!G}(u,v) = 1 - \frac{ \text{\emd}_{\!\!G}(u,v) } { \dist_G(u,v) }
> 
1 - \frac{1}{2} - \frac{3}{m'+2}
>0
$.

\medskip
\noindent
\textbf{Proof of soundness}

\smallskip
\noindent
Let $E'$ be the new set of edges added and 
let 
$G'=(V,E\cup E',w')$ be the new graph obtained after the insertion of edges
such that 
$\text{\ric}_{\!\!G'}(u,v)>0$.
Since a trivial valid solution can be obtained by connecting $u$ to each of the set nodes
via $m$ edges, we may assume that 
$|E'|\leq m<n^{0.1}$.

\begin{lemma}\label{lem-no-heavy}
No dirt shipment pertaining to the calculation of 
$\text{\emd}_{\!\!G'}(u,v)$ 
can use a heavy edge.
\end{lemma}

\begin{proof}
By Proposition~\ref{prop-rationalility}
any amount of shipment in a
calculation of the \emd\ in $G'$ must be an integral multiple of $\nicefrac{1}{q}$ 
where $q=\mathrm{LCM}(1+\deg_{G'}(u),1+\deg_{G'}(v))$.
Note that both 
$\deg_{G'}(u)+1$
and 
$\deg_{G'}(v)+1$ can be at most 
$n+m+1$, thus 
$q\leq (n+m+1)^2<n^3={n_1}^{30c}$.
Suppose that a dirt shipment from some node in $V_u$ to another node in $V_v$ 
uses a heavy edge. 
Then the cost of this transportation is at least $\frac{W}{q}$ 
and thus 
$
\text{\emd}_{\!\!G'}(u,v) 
>
\frac{W}{q}\geq \frac{{n_1}^{100c}}{{n_1}^{30c}}>2
\,\Rightarrow\,
\text{\ric}_{\!\!G'}(u,v) = 1 - \frac{ \text{\emd}_{\!\!G'}(u,v) } { 2 } <0
$, a contradiction.
\end{proof}

\begin{lemma}\label{lem-sound}
Let $V_{\cS}$ be the collection of set nodes 
that are reachable from $u$ or $v$ without using any heavy edges.
Then, the collection of sets 
$\cT'=\{
\cS_j \,|\, v_{\cS_j}\in V_{\cS}
\}$
form a valid solution of the \msc.
\end{lemma}

\begin{proof}
Let $|\cT'|=m'$  and assume without
loss of generality that $\cT'=\{\cS_1,\dots,\cS_{m'}\}$.
Let $E''\subseteq E'$ be the set of inserted edges whose both endpoints
are from the set of nodes  
$\cT'\cup \{u,v\}$.
Suppose that $\cT'$ is not a valid solution of the \msc.
Thus, there exists $N={n_1}^{10c}$ element nodes 
$v_{a_{j,1}},\dots, v_{a_{j,N}}$
such that 
$v_{a_{j,1}},\dots, v_{a_{j,N}}
\notin \cup_{j=1}^{m'}\cS_j
$.
Let $N'$ be the new inserted edges connecting $u$ to $N'$ 
of these element nodes, say 
$v_{a_{j,1}},\dots, v_{a_{j,N'}}$; 
note that $N'\leq |E'|\leq m={n_1}^{c}$.
Then
the $N-N'$ nodes 
$v_{a_{j,N'+1}},\dots, v_{a_{j,N}}$ 
do not belong to the set $V_u$ in $G'$.
Since by 
Lemma~\ref{lem-no-heavy}
shortest paths from nodes in $V_u$ to the 
nodes 
$v_{a_{j,N'+1}},\dots, v_{a_{j,N}}$ 
cannot use a heavy edge, 
for every node 
$v_{a_{j,\ell}}$
for $\ell\in\{N'+1,\dots,N\}$ 
there must be a node in $V_u$ such that the last edge in
the shortest path from that node to 
$v_{a_{j,\ell}}$
must be a newly inserted edge of one of the following form: 
\textbf{(\emph{i})}
$\{v_{a_{j,\ell}},
v_{\cS_k}
\}$
for some
${\cS_k}\in\cT'$ or 
\textbf{(\emph{ii})}
$\{v_{a_{j,\ell}},
v_{a_{j',\ell'}}
\}$
for another element node 
$v_{a_{j',\ell'}}$.
Thus, 
$|E'| > \frac{N-N'}{2}
\geq
\frac{{n_1}^{10c} - {n_1}^{c}}{2}
$,
contradiction the fact that 
$|E'|\leq
{n_1}^c
$.
\end{proof}

By Lemma~\ref{lem-sound}
we have 
$|E'|\geq |E''|
\geq m'
\geq 
\opt_{\!\!\text{\footnotesize\msc}}\geq (\eps\ln (n+1)) \xi
=
(\eps\ln \deg_G(v) ) \xi
$.
\end{proof}

\section{Conclusions and Future Research}

This article 
provides a formal framework for studying algorithmic and computational complexity issues 
for detecting critical edges in an undirected graph using Ollivier-Ricci curvatures. 
Since there are $12$ non-trivial version of problems and we have covered only a few of them,
there is considerable opportunity for studying algorithmic aspects of the remaining versions.
As for improvements for results in this article, 
curiously the exact perfect matching problem has a polynomial-time 
deterministic algorithm if the matching is allowed to differ just by one edge 
from perfect matching~\cite{Yus12}
and one possible future research work may involve 
de-randomizing
our randomized algorithms
using the ideas/results in~\cite{Yus12}.

Another possible future direction of research could be to redefine the curvature 
measure by using another distance between two probability distribution instead
of \emd. The simplest such distance which can be easily computed is the \emph{total variation 
distance}, however this distance would not be very useful in practice. Note that 
distances which are not symmetric (\EG, the Kullback-Leibler divergence)
are not be appropriate for usage in curvature measure.

\bibliographystyle{plainurl}
\bibliography{references-journ}

\appendix

\begin{center}
\textbf{Appendix}
\end{center}

\section{Proof of Theorem~\ref{thm-folklore}}
\label{sec-app-folklore}

%
\noindent
\textbf{(\emph{a})}
Given a valid solution 
$\vec{z}=\{z_{u_i,v_j}\,|\, u_i\in L, v_j\in R\}$ 
of value 
$\text{\tco}_{\!\!H}(\vec{z},L,R,\PP_L,\PP_R)$
on $H$, 
we obtain a valid solution of same cost on 
$\widehat{H}$ by setting 
$\widehat{z}_{u_{i,k},v_{j,\ell}}=\frac{z_{u_i,v_j}}{ab}$ for $k\in\{1,\dots,\abig\},\ell\in\{1,\dots,\bbig\}$.
Since 
$
\sum_{k=1}^a\sum_{\ell=1}^b \widehat{z}_{u_{i,k},v_{j,\ell}}=z_{u_i,v_j}
$
it follows that 
$\text{\tco}_{\!\!\widehat{H}}(\vec{\widehat{z}},\widehat{L},\widehat{R},\widehat{\PP_L},\widehat{\PP_R)}
=$
\newline
$\text{\tco}_{\!\!H}(\vec{z},L,R,\PP_L,\PP_R)$.
That our solution for 
$\widehat{H}$ is indeed a valid solution follows from the following equalities:
\begin{gather*}
\sum_{j=1}^s\sum_{\ell=1}^b \widehat{z}_{u_{i,k},v_{j,\ell}}
=
\sum_{j=1}^s\sum_{\ell=1}^b \frac{z_{u_i,v_j}}{ab}
=
\sum_{j=1}^s \frac{z_{u_i,v_j}}{a}
=
\frac{1}{ar} = \frac{1}{q}
\\
\sum_{i=1}^r\sum_{k=1}^a \widehat{z}_{u_{i,k},v_{j,\ell}}
=
\sum_{i=1}^r\sum_{k=1}^a \frac{z_{u_i,v_j}}{ab}
=
\sum_{i=1}^r \frac{z_{u_i,v_j}}{b}
=
\frac{1}{bs} = \frac{1}{q}
\end{gather*}
Conversely, 
given a valid solution 
$\vec{\widehat{z}}=\{\widehat{z}_{u_{i,k},v_{j,\ell}}\,|\, u_{i,k}\in \widehat{L}, v_{j,\ell}\in \widehat{R}\}$ 
of value 
$\text{\tco}_{\!\!\widehat{H}}(\vec{\widehat{z}},\widehat{L},\widehat{R},\widehat{\PP_L},\widehat{\PP_R)}$
on $\widehat{H}$, 
we obtain a valid solution of same cost on 
$H$ by setting 
$
z_{u_i,v_j} =
\sum_{k=1}^a \sum_{\ell=1}^b \widehat{z}_{u_{i,k},v_{j,\ell}}
$. 
It is obvious that 
$\text{\tco}_{\!\!H}(\vec{z},L,R,\PP_L,\PP_L)
=
\text{\tco}_{\!\!\widehat{H}}(\vec{\widehat{z}},\widehat{L},\widehat{R},\widehat{\PP_L},\widehat{\PP_R)}
$.
That our solution for 
$H$ is indeed a valid solution follows from the following equalities:
\begin{gather*}
\sum_{j=1}^s z_{u_i,v_j}
=
\sum_{j=1}^s \sum_{k=1}^a \sum_{\ell=1}^b \widehat{z}_{u_{i,k},v_{j,\ell}}
=
\sum_{k=1}^a \sum_{j=1}^s \sum_{\ell=1}^b \widehat{z}_{u_{i,k},v_{j,\ell}}
=
\sum_{k=1}^a \frac{1}{q}
=
\frac{a}{q} = \frac{1}{r}
\\
\sum_{i=1}^r z_{u_i,v_j}
=
\sum_{i=1}^r \sum_{k=1}^a \sum_{\ell=1}^b \widehat{z}_{u_{i,k},v_{j,\ell}}
=
\sum_{\ell=1}^b \sum_{i=1}^r \sum_{k=1}^a \widehat{z}_{u_{i,k},v_{j,\ell}}
=
\sum_{\ell=1}^b \frac{1}{q}
=
\frac{b}{q} = \frac{1}{s}
\end{gather*}

\medskip
\noindent
\textbf{(\emph{b})}
Let $e_1,\dots,e_{q^2}$ be 
any arbitrary ordering of the edges of 
$\widehat{H}$.
Suppose that we multiply the amount of dirt at each source node and the storage requirement of each destination 
node of $\widehat{H}$ by $q$. This increases the value of 
$\text{\emd}_{\!\!\widehat{H}}(\widehat{L},\widehat{R},\widehat{\PP_L},\widehat{\PP_R})$
precisely by a factor of $q$ and makes each source dirt amount and each destination storage requirement exactly $1$.
Then the constraints of the linear program in \FI{f1}
corresponding to this modification of the computation of 
$\text{\emd}_{\!\!\widehat{H}}(\widehat{L},\widehat{R},\widehat{\PP_L},\widehat{\PP_R})$
can be alternatively written in matrix-vector notation as 
$\mathbf{A}\mathbf{e}=\mathbf{1}$
where $\mathbf{1}$ is a $2q\times 1$ column vectors of all $1$ entries, 
$\mathbf{e}$ is the column vector $\big(y_{e_1} y_{e_2} \dots y_{e_{q^2}}\big)^{\mathrm{T}}$ for the variables,
and $\mathbf{A}=[a_{i,j}]$ is a $2q\times q^2$ node-edge incidence matrix with entries in its $k\tx$ column 
corresponding to the edge 
$e_k=\{u_i,v_j\}$ 
given by $a_{\ell,k}=1$ if $\ell=i$ or $\ell=j$ 
and 
$a_{\ell,k}=0$ otherwise. 
The matrix $\mathbf{A}$ is totally unimodular~\cite[Theorem 13.2]{PS82}
and since the entries of $\mathbf{1}$ are all integers,
there is an optimal solution of this linear program 
in which all entries of $\mathbf{e}$ are $0$ or $1$~\cite[Theorem 13.1]{PS82}.
Since each source dirt amount and each destination storage requirement is exactly $1$, 
the set of edges $\{e_p \,|\, y_{e_p}=1\}$
in this optimal solution
form a perfect matching in $\widehat{H}$ of total weight
$q \times \text{\emd}_{\!\!\widehat{H}}(\widehat{L},\widehat{R},\widehat{\PP_L},\widehat{\PP_R})$.
Setting 
$\widehat{z}_{u_{i,k},v_{j,\ell}}=\frac{y_{e_p}}{q}$ for each edge $e_p=\{u_{i,k},v_{j,\ell}\}$
we get an optimal solution of value 
$\text{\emd}_{\!\!\widehat{H}}(\widehat{L},\widehat{R},\widehat{\PP_L},\widehat{\PP_R})$
on $G$.

\medskip
\noindent
\textbf{(\emph{c})}
Given any perfect matching $\M$
of $\widehat{H}$ 
of cost $\alpha=\sum_{e\in\M}\widehat{w}(e)$,   
construct a valid solution on 
$\widehat{H}$ 
of cost $\nicefrac{\alpha}{q}$ by setting   
$\widehat{z}_{u_{i,k},v_{j,\ell}}=\nicefrac{1}{q}$
if 
$\{u_{i,k},v_{j,\ell}\}\in\M$ and 
$\widehat{z}_{u_{i,k},v_{j,\ell}}=0$
otherwise.

\end{document}